\documentclass[Afour,sageh,times]{sagejarxiv}

\usepackage{moreverb,url}

\usepackage[colorlinks,bookmarksopen,bookmarksnumbered,citecolor=red,urlcolor=red]{hyperref}
\usepackage{csquotes}
\usepackage{natbib}
\usepackage{listings}

\usepackage[font=footnotesize]{subfig}

\newcommand{\eqnsplit}[1]{\begin{equation*}\begin{split}
#1\end{split}\end{equation*}}

\usepackage{amsthm}
\usepackage{thmtools,thm-restate}

\usepackage{todonotes}

\usepackage{fancyvrb}

\newcommand\BibTeX{{\rmfamily B\kern-.05em \textsc{i\kern-.025em b}\kern-.08em
		T\kern-.1667em\lower.7ex\hbox{E}\kern-.125emX}}

\newcommand{\rem}[1]{}

\usepackage{etoolbox}

\usepackage[TS1,T1]{fontenc}

\begin{document}

\runninghead{}


\title{A dynamical model of platform choice and online segregation}

\author{Sven Banisch\affilnum{1}\affilnum{*}, Dennis Jacob\affilnum{2}, Tom Willaert\affilnum{3}, Eckehard Olbrich\affilnum{4}}
\affiliation{\affilnum{1} Institute of Technology Futures, Karlsruhe Institute of Technology, Germany\\
\affilnum{2} EECS Computer Science Division, UC Berkeley, Berkeley, CA, United States\\
\affilnum{3} Brussels School of Governance, Vrije Universiteit Brussel, Brussels, Belgium \\
\affilnum{4} Max Planck Institute for Mathematics in the Sciences, Leipzig, Germany\\
\affilnum{*} Corresponding author (sven.banisch@universecity.de)}

\begin{abstract}
In order to truly understand how social media might shape online discourses or contribute to societal polarization, we need refined models of platform choice, that is: models that help us understand why users prefer one social media platform over another.
This study develops a dynamic model of platform selection, extending Social Feedback Theory by incorporating multi-agent reinforcement learning to capture how user decisions are shaped by past rewards across different platforms. A key parameter ($\mu$) in the model governs users' tendencies to either seek approval from like-minded peers or engage with opposing views. Our findings reveal that online environments can evolve into suboptimal states characterized by polarized, strongly opinionated echo chambers, even when users prefer diverse perspectives. Interestingly, this polarizing state coexists with another equilibrium, where users gravitate toward a single dominant platform, marginalizing other platforms into extremity. Using agent-based simulations and dynamical systems analysis, our model underscores the complex interplay of user preferences and platform dynamics, offering insights into how digital spaces might be better managed to foster diverse discourse.
\end{abstract}


\maketitle



\section{Introduction}


Digital platforms have a profound impact on society, influencing how people communicate, access information, and express opinions \citep{boyd2014complicated}. While these platforms have become integral to daily life, they are also increasingly scrutinized for their role in fostering polarization, the spread of misinformation, and growing mistrust in democratic institutions \citep{sunstein2018republic, tucker2018polarization}. The recent turmoil surrounding Twitter (now X) 
and the emergence of conglomerates like Meta—encompassing Facebook, Instagram, WhatsApp, and Threads—has marked a new phase of public skepticism toward social media platforms, their leadership, and the algorithms driving user experiences. 
In response, users are becoming more discerning, gravitating toward platforms that better align with their personal communication preferences \citep{brady2023social}.


One of the primary concerns surrounding digital platforms is their role in reinforcing polarization and creating echo chambers, where users predominantly engage with like-minded individuals \citep{sunstein2018republic}. This selective exposure not only amplifies existing opinions but also fosters ideological divides and undermines cross-cutting discourse \citep{stroud2010polarization}.
A key driver of this phenomenon lies in the cultural norms and values embedded within each platform, which shape the kinds of opinions that are promoted or marginalized. Platforms differ in their normative structures -- ranging from rigid moderation policies to more open, decentralized environments -- resulting in distinct platform cultures \citep{scharlach2023value,raemdonck2022conceptueel}. Users actively navigate these environments, often selecting online spaces that align with their own communicative preferences, values, and even political identities \citep{garrett2009echo,knobloch2015choice}. At the same time, the market dynamics of platform ecosystems enable the constant emergence of new platforms, giving users the option to switch when a platform's technical features or culture no longer suit their needs. This fluid platform ecosystem raises important questions about how user choices, shaped by social feedback and platform-specific norms, contribute to the formation of either polarized spaces or diverse, open environments.


In particular, by means of their specific technical and social affordances, social media platforms might support key functions of the public sphere in liberal democracies. They might (1) facilitate the circulation of information and provide public access to information, (2) facilitate rational-critical debate among the public, (3) foster the creation and expression of collective identities, and (4) allow actors to coordinate and collectively perform actions \citep{willaert2024social}. When choosing which platforms to actively engage with, users might thus seek out those social media that best support the functions of the public sphere they value most highly. A user interested in belonging to a specific group or collective, might for instance look for those platforms that host like-minded individuals. On the other hand, users who wish to challenge their assumptions and expose themselves to new information, or wish to engage in rational debate with others, might be drawn to platforms that offer a diversity of perspectives and opinions. This dynamic selection process highlights the interplay between user motivations and platform affordances, underscoring how individual choices can shape the broader online discourse and, potentially, the structure of the digital public sphere.

To better understand the complex interplay between user behavior and online platform ecologies, this paper extends social feedback theory (SFT) by incorporating reinforcement learning (RL) to model adaptive user-platform interactions \citep{Banisch2019opinion,Banisch2022modeling}. In SFT, user behavior is shaped by the social rewards or feedback individuals receive from others, making it an ideal framework for understanding how social media platforms influence opinion expression and choice. From \cite{jacob2023polarization} we borrow an explicit representation of virtual environments. We refer to these different environments as platforms throughout the paper noticing that they might also relate to online spaces within a single large social network service. By integrating RL, we allow users to adapt their platform choices based on the rewards they have experienced, such as the affirmation of like-minded peers or exposure to opposing views. This approach captures the dynamic, evolving nature of online environments, where user preferences for different platforms co-evolve with the composition of the platform ecology (see Figure \ref{fig:coevo}). In contrast to models that assume static decision-making \citep{granovetter1978threshold,schelling1971dynamic}, our approach emphasizes how feedback loops between many user actions and platforms can lead to different system-wide outcomes -- ranging from polarized echo chambers to diverse, open spaces -- depending on the balance between social approval and diversity-seeking behaviors. In this way, our paper provides a more focused understanding of how 
platform markets and user behavior influence the evolution of polarized environments or more inclusive spaces.

\begin{figure}[ht]
	\centering
	\includegraphics[width=0.99\linewidth]{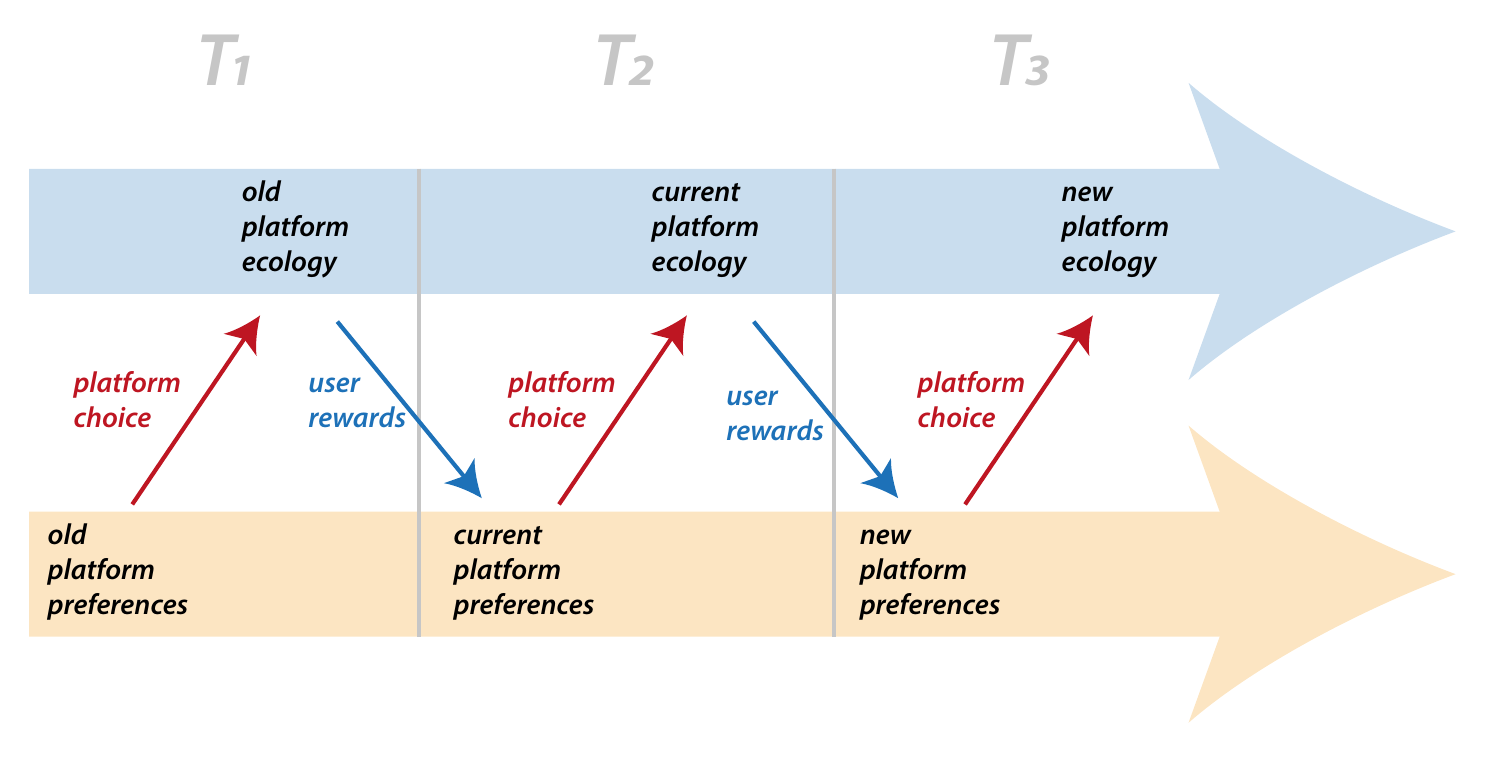}
	\caption{Platform choice models address the co-evolutionary process of users’ preferences and decisions to engage with different platforms, and the resulting platform ecology in which these decisions take place.}
	\label{fig:coevo}
\end{figure}

The model simulates a population of users, each holding one of two opposing opinions, who choose between a set of platforms for expressing their views. Users are influenced by two primary rewards: social approval, which is derived from interacting with like-minded individuals, and diversity, which comes from engaging with opposing views. Each platform is characterized by the proportion of users holding each opinion and the level of activity on the platform, which affects user satisfaction. A key parameter $\mu$ governs the relative importance users place on social approval versus diversity. Through RL, users adapt their platform choices over time, updating their platform preferences based on the rewards they experience. This leads to dynamic changes in platform composition as users gravitate toward platforms that maximize their satisfaction. As the simulation evolves, we observe the formation of two distinct platform configurations. On the one hand, the simulation forms polarized echo chambers, where users cluster around like-minded peers.
On the other hand, we see the emergence of a single very large platform where different opinions co-exist and other strongly opinionated platforms are marginal. Remarkably, both configurations are stable in a specific parameter range $\mu$.
Our model thus captures the feedback loops between individual behavior and platform dynamics, revealing how small changes in user preferences or platform features can lead to vastly different systemic outcomes.

These results echo a key insight from Schelling’s seminal work on segregation \citep{schelling1969models,schelling1971dynamic}, which demonstrated how individual motivations can aggregate into unintended macro-level outcomes. In Schelling’s model, even mild preferences for like-minded neighbors led to sharp spatial segregation, even though individuals did not explicitly seek such extremes. Like Schelling’s work, our model highlights how individual motivations -- in our case, for social approval or diversity -- aggregate to produce collective outcomes such as polarized platforms or diverse spaces. The model parameter, $\mu$, governs the balance between social approval and diversity rewards for the individual users. When $\mu < 0.5$, polarization is the dominant outcome as users prioritize social approval. As $\mu$ increases beyond 0.5, the system favors global diversity as the optimal outcome, yet polarization persists as the only stable outcome for a range of $\mu$ values between 0.5 and 0.64. 
Strong motivations for diversity  ($0.64 < \mu < 0.875$) eventually favor a globally optimal equilibrium, where a single mega-platform features both opinion groups and the remaining platforms are marginalized to extremity. Nevertheless, even in this setting suboptimal equilibria of online segregation remain a stable outcome of the model (for $\mu < 0.79$).
These findings underscore how feedback loops between user behavior and platform characteristics can drive polarization, even when users genuinely value diversity, recalling Schelling’s demonstration of how small individual preferences for like-minded others can lead to unexpected systemic segregation.

This paper makes several key contributions to the study of platform dynamics and user behavior on social media. First, this paper extends SFT to model how users adapt their communication choices in a dynamic platform environment. Unlike most models of polarization that focus on opinion dynamics within fixed networks \citep{flache2017models,Banisch2019opinion,lorenz2021individual}, our model explicitly captures how users select platforms based on social feedback. By simulating the evolving interplay between user preferences and platform features, the model provides a novel framework for studying how online segregation emerges from individual communication decisions.

Second, we conduct a rigorous mathematical analysis of the "tipping points" in the model using bifurcation analysis. This approach identifies critical transitions in system behavior, specifically how variations in the balance between social approval and diversity-seeking behavior (governed by parameter $\mu$) lead to different platform outcomes -- ranging from polarized echo chambers to diverse spaces. This analysis offers new insights into the parameter regimes where polarization persists, and where more diverse, inclusive environments can emerge, providing a theoretical basis for understanding platform segregation and diversity dynamics.

Third, the model is highly flexible and can accommodate a wide range of attributes, beyond political or ideological stances. "Opinion" in the model can refer to cultural values, norms, or even preferences regarding communication standards, such as the toxicity of language. Furthermore, the model can be applied to posting decisions within platforms like Reddit, where users choose between different communities (called "subreddits") based on feedback such as upvotes or comments. This flexibility makes the modeling framework presented in this paper a versatile tool for fine-tuning and calibrating models to specific online settings and for exploring how platform design and user incentives interact across different contexts. 
In this paper, we focus on rigorous theoretical understanding of a simple model with fixed binary opinions and homogeneous motivations $\mu$ to explain basic dynamics of online segregation. Future model iteration attuned to real platforms will provide actionable insights for platform designers and policymakers aiming to foster more inclusive digital spaces.

\section{Agent-Based Model}
\label{sec:binarymodel}


In the model, we assume that a numbers of $N$ users (called agents) decide to engage in opinion expression activity on $M$ different social media platforms. Different platforms are indexed as $k \in \{0,1,2,\ldots M \}$. We further assume that agents hold two opposing opinions ($o_i \in \{-1,1\}$) on a single issue. These opinions do not change in the dynamical process, dividing the population into $N_+$ agents holding opinion 1 and $N_-$ agents with opinion -1. We may hence also think of the $o_i$'s as a cultural attribute or identity in the spirit of \cite{schelling1971dynamic} or \cite{tornberg2021modeling}.


In subsequent rounds, agents choose to express their opinion on their currently preferred platform and evaluate that platform based on the social feedback obtained within it (called rewards). Agents can also refrain from opinion expression, if they prefer remaining silent over choosing to voice their opinion on one of the platforms available to them. This first "environment" is indexed by $0$ capturing the decision to be silent in the current round. In the next paragraphs, we discuss the model steps one by one.

\paragraph{Platform choice.} 
Following \cite{Banisch2019opinion}, our model uses reinforcement learning (RL) to model adaptive agent behavior. Specifically, in this model, preferences for platforms (or silence) are represented via subjective expected utilities (i.e., $Q$-values) associated to the different options among which agents can choose.
These evaluations are updated round by round based on past experiences (i.e., rewards). Let \( Q_i(k) \) denote the Q-value agent \( i \) associates with platform \( k \). Agents (indexed by \( i \)) choose platform \( k \) with a probability derived from the Boltzmann soft-max procedure, a method commonly used in decision-making models where options with higher expected rewards are chosen with higher probabilities. This approach balances exploration and exploitation, allowing agents to occasionally choose less-optimal platforms to explore new possibilities. The probability of choosing platform \( k \) is given by:
\begin{equation}
    p_k^i = \frac{e^{\beta Q_i(k)} }{\sum_{l = 0}^M e^{\beta Q_i(l)}}.
    \label{eq:actionselection}
\end{equation}
The parameter \( \beta \), known as the exploration rate, governs the balance between randomness and rationality in decision-making. A low value of \( \beta = 0 \) means that agents make random choices, while a high \( \beta = \infty \) means that agents always select the option with the highest expected reward (fully rational choice).
Recall that the first "environment" $k=0$ represents the option of refraining from opinion expression, and we assume that $Q_i(0) = 0$ for all agents throughout the entire process. Platform choice $c_i$ is governed by Eq. (\ref{eq:actionselection}) such that $c_i = k$ with probability $p_k^i$.

\paragraph{Platform activity and opinion.} After all agents made their choices, we compute the activity $A_k$ and opinion $O_k$ of each platform. 
The activity of platform $k$ is defined as the number of users expressing their opinion on platform \( k \) in a given round
\begin{equation}
    A_k = \sum_{i = 1}^{N} \delta(c_i = k)
    \label{eq:platformactivity}
\end{equation}
where \( \delta(c_i = k) \) is an indicator function that equals 1 if agent \( i \) chose platform \( k \), and 0 otherwise.

The overall opinion \( O_k \) on platform \( k \) represents the average opinion expressed by users on that platform. It is computed as the mean opinion of all active users on the platform in a given round:
\begin{equation}
    O_k = \frac{1}{A_k} \sum_{i = 1}^{N} \delta(c_i = k) o_i
    \label{eq:platformopinion}
\end{equation}
where $o_i$ represents the opinion of agent $i$.
Platform activity and opinion are the most central variables for analyzing the temporal evolution of the ABM. 


\paragraph{Rewards.} 


In the model, users have two distinct motivations that make platform engagement appealing: (a) social approval of their expressed opinions and (b) exposure to opinion diversity on the platform. These incentives reflect the dual motivations that align with broader public sphere functions, where users may either seek affirmation within like-minded groups or pursue diverse perspectives. We represent these motivations through two types of rewards. The first reward (a) assumes that agents perceive social support for their opinion as positive reinforcement. The average opinion on the chosen platform \(O_k\) reflects the likelihood of an agent receiving affirmative responses, capturing this social feedback as follows:
\begin{equation}
    r_a^i(c_i = k) = O_k o_i.
    \label{eq:rewardapproval}
\end{equation}
Accordingly, when the average platform opinion \(O_k\) is negative, an agent with a negative opinion will receive positive feedback, while an agent with \(o_i = +1\) will experience negative feedback. 


But agents are motivated not only by support and the affirmation of seeing their opinions echoed. A platform might also lose appeal if there is insufficient opinion diversity, as a lack of contrasting views reduces the potential for argumentation and meaningful discussion. This second type of reward (b) represents the value agents place on encountering diverse perspectives, and is captured by
\begin{equation}
    r_d^i(c_i = k) = 1 - |O_k|,
\end{equation}
which decreases linearly with the opinionatedness of the platform in the current round. This diversity reward reaches a maximum of one when the platform consists of a balanced mix of agents from both opinion groups and drops to zero if all agents share the same opinion. In this way, $r_d$ reflects the reward agents gain from a diverse online environment in terms of an increasing reward by exposure to the opinion of others.


We combine these two motivations introducing the basic model parameter $\mu$ by
\begin{equation}
    r_k^i = (1-\mu) r_a^i + \mu r_d^i.
    \label{eq:rewards}
\end{equation}
The parameter \( \mu \) models the tradeoff between seeking social approval and opinion diversity as a convex combination of the two reward contributions $r_a$ and $r_d$. When \( \mu = 0 \), agents care exclusively about social approval (homophily), while when \( \mu = 1 \), they prioritize engaging with diverse perspectives. Intermediate values of \( \mu \) capture a balanced social preference for both approval and diversity.

\paragraph{Agent update.} The basic idea of social feedback models is that agents learn from the rewards obtained in one round and adjust their Q-values for the next round accordingly. This can be realized by Q-learning where the value $Q(k)$ of a chosen platform is adjusted by
\begin{equation}
    Q_i^{t+1}(k) = Q_i^{t}(k) + \alpha ( r_k^i - Q_i^{t}(k) ).
    \label{eq:Qlearning}
\end{equation}
That is, the new evaluation $Q_i^{t+1}(k)$ is adjusted by the difference between the experienced reward $r_k^i$ and the current evaluation $Q_i^{t}(k)$, typically referred to as temporal difference error. 
The learning rate \( \alpha \) determines how strongly agents update their evaluation of a platform based on recent experiences. A low \( \alpha \) means that agents are slow to adapt their preferences, while a high \( \alpha \) makes them more responsive to new feedback.

\paragraph{Platform rewards.}
In addition to the platform activity $A_k$ and opinion $O_k$ \eqref{eq:platformactivity} and \eqref{eq:platformopinion}, we also keep track of the rewards agents obtain on average from a given platform. In analogy to the previous observable, we define the platform reward as
\begin{equation}
    R_k = \frac{1}{A_k} \sum_{i = 1}^{N} \delta(c_i = k) r_k^i
    \label{eq:platformreward}
\end{equation}
as the average reward agents that engage with platform $k$ receive in the current round. The platform rewards thus trace how satisfied users are on a certain platform and reflect the extent to which a platform market serves the needs of users.

\section{Simulations}

\subsection{Model phenomenology}

\paragraph{Simulation setting.} Round by round, all agents choose one out of $M$ platforms based on an evaluation $Q_i(k)$. For the simulation within this section, we use $N = 500$ agents that can chose among $M = 5$ different platforms. The simulation is run for 5000 rounds with the parameters $\alpha = 0.1, \beta = 8$ fixed. At start, all Q-values are set to 0 so that all platforms and silence are chosen with equal probability. Opinions $o_i \in \{-1,1\}$ are assigned at random so that, on average, half the population supports one and the other half the other opinion ($N_+ \approx 250, N_- \approx 250$). Our variable of interest is the preference for diversity $\mu$.

\paragraph{Preference for positive feedback only ($\mu = 0$).}
Let us first consider the extreme cases, starting with agents who derive reward from seeing their opinion supported on the platform and do not strive for diversity. This corresponds to $\mu = 0$ in Eq. (\ref{eq:rewards}). In Fig. \ref{fig:M00} we show an exemplary run for this case. The evolution of the normalized platform activity for the five platforms ($A_k(t)/N$) is shown in the top panel, and the platform opinion $O_k(t)$ is shown below. The proportion of silent agents is also shown in the activity plot by the black curve (very close to zero).

\begin{figure}[ht]
	\centering
	\includegraphics[width=0.99\linewidth]{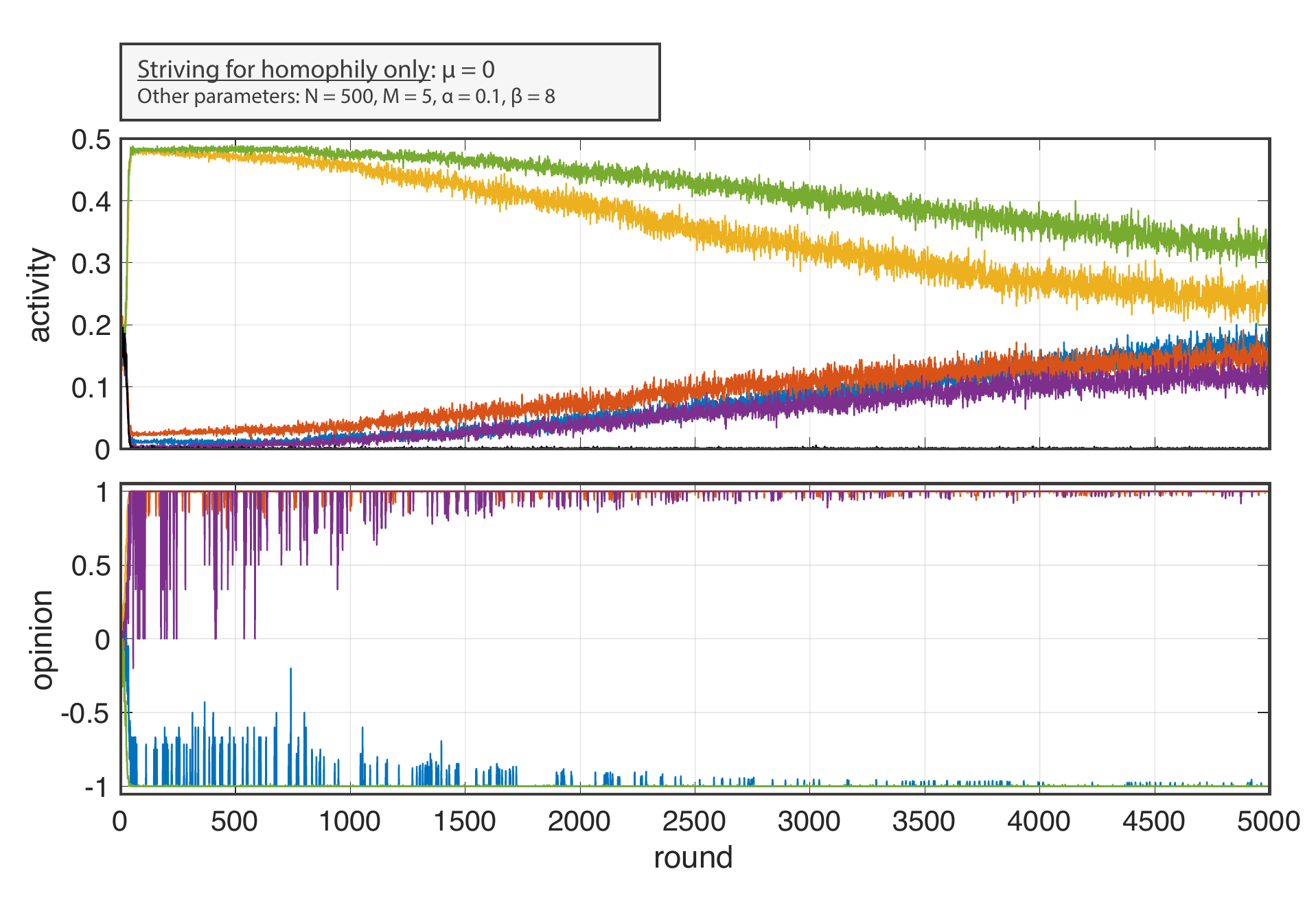}
	\caption{Example run for $\mu = 0.0$ which shows platform polarization due to homophily, with opinions clustering on separate platforms. The evolution of the normalized platform activity for the five platforms (\( A_k(t)/N \)) is shown in the top panel, and the platform opinion (\( O_k(t) \)) is shown below. The proportion of silent agents is also shown in the activity plot by the black curve (very close to zero). Other parameters: $N = 500, M = 5, \alpha = 0.1, \beta = 8$.}
	\label{fig:M00}
\end{figure}

In this setting, the system quickly approaches a state where agents with a positive opinion meet on three platforms (yellow, orange and violet) and agent with a negative opinion gather on the remaining two (green and blue). We shall refer to this as platform segregation or polarization, as agents with opposing opinions have a very low chance of interacting on one and the same platform. Notice in the activity plot that two platforms initially attract a large user share so that almost all agents prefer one of the two. In particular, agents with a negative stance gather on the green platform and agent with $o_i = 1$ on the yellow one. The remaining platforms, though strongly opinionated from the beginning, are very low in activity, but they attract more and more users as time evolves. 
Eventually, in this run, platform activity will converge to a stationary profile in which agents with the same opinion choose the respective platforms on their side with equal probability ($\approx 1/4$ for green and blue, and $\approx 1/6$ for yellow, orange and violet). The convergence of Q-values causes agents to choose among platforms that share their opinion uniformly at random.
Finally, the number of silent agent (i.e. the probability to refrain from opinion expression) is very low.

Notice further that we may observe model realizations in which only a single platform emerges on one side of the opinion spectrum, and the other four share users on the other side (not shown). The activity then approaches a state in which half on the users (i.e. those with the respective opinion) gather on one platform, whereas the other opinion group engages at equal share on the remaining ones ($\approx 1/8$).

\paragraph{Striving for diversity only ($\mu = 1$).}
The second limiting case is when agent only opt for diversity, and draw no reward from seeing their own opinion backed by others. That is, $\mu = 1$. An example run of this situation is shown in Fig. \ref{fig:M10}.

\begin{figure}[ht]
	\centering
	\includegraphics[width=0.99\linewidth]{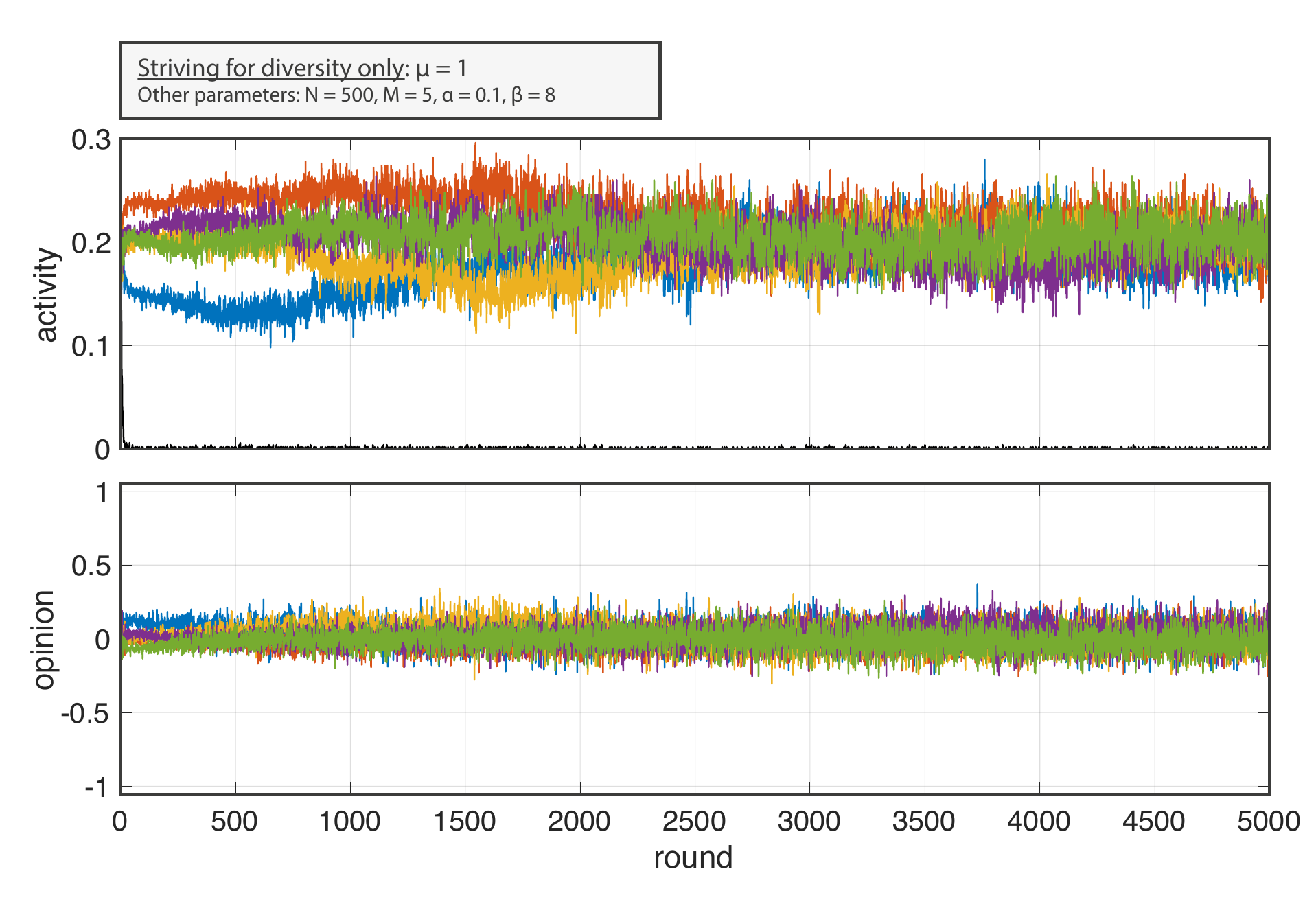}
	\caption{Example run for $\mu = 1.0$, illustrating balanced platform activity with diverse, non-opinionated platforms when agents prioritize diversity. In this parameter setting, all platforms approach an average opinion \( O_k \approx 0 \) and are hence not opinionated. Other parameters: $N = 500, M = 5, \alpha = 0.1, \beta = 8$.}
	\label{fig:M10}
\end{figure}

In this parameter setting, we observe that all platforms approach an average opinion $O_k \approx 0$ and are hence not opinionated. The activity approach a value of $A_k/N \approx 1/5$ meaning that agents do not strongly prefer one platform over another. Again, we remark that the Q-values are approximately equal for all $k$ options so that the platform choice probabilities are equal. The number of silent agents is very low.

\paragraph{Striving for diversity and homophily ($\mu = 0.7$).}
A series of interesting phenomena may emerge in between these two limiting cases when mixing rewards from positive feedback (affirmation) and diversity by the parameter $\mu \in [0,1]$. We illustrate this by looking at a parameter value $\mu = 0.7$ by which agents draw a considerable reward from interaction with agents from the other group, but also value approval. 
In Fig. \ref{fig:M07polar} and \ref{fig:M07mega} we compare two model runs with $\mu = 0.7$. We show the first 1000 rounds of the model. 


The parameter $\mu = 0.7$ is chosen because the model can enter two very different equilibrium states. Even under the same initial assignment of Q-values (all $Q_k^0 = 0$), the stochastic process implemented by the model may either approach a state of platform polarization (as with $\mu = 0$), or a state in which one maximally diverse platform dominates over four marginal but opinionated platforms. 

\begin{figure}[ht]
	\centering
	\includegraphics[width=0.99\linewidth]{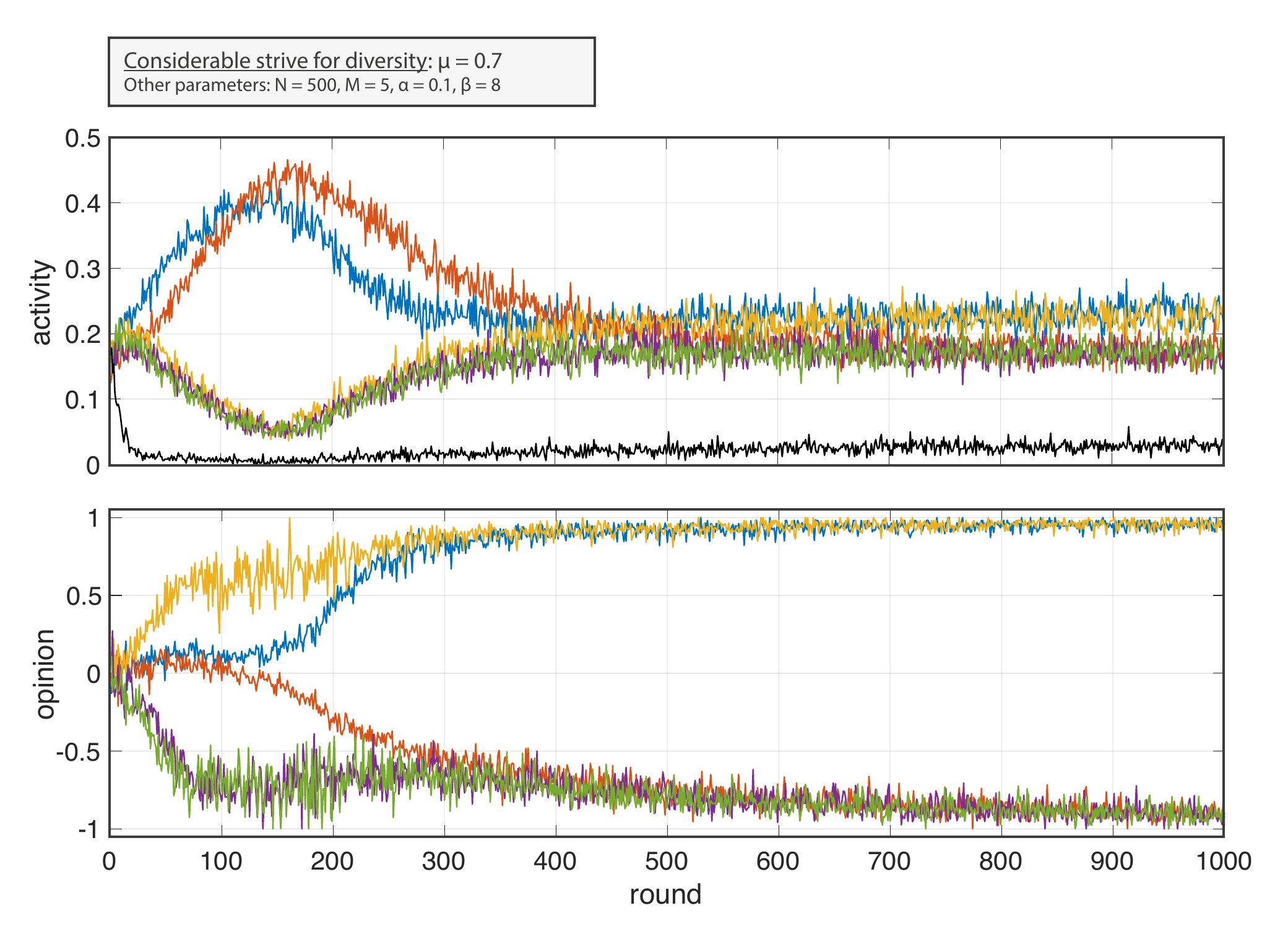}
	\caption{First run for $\mu = 0.7$, demonstrating polarization after an initial phase of platform competition (referred to as Run A). Other parameters: $N = 500, M = 5, \alpha = 0.1, \beta = 8$.}
	\label{fig:M07polar}
\end{figure}

In the first realization (run A), platform polarization is observed after an initial phase in which two platforms compete to provide a diverse space. 
For up to 150 rounds, the orange and the blue platform remain in the middle of the opinion scale. This means that an approximately equal share of users from both opinion groups express their opinion on them.
These platforms are also the most active during that phase.
But, in this run, this situation becomes unstable so that one of the platforms is drawn to the positive and the other one to the negative side.
After 400 rounds, no platform affords opinion diversity and the activity levels converge. In this run, two platforms (yellow and blue) are populated by agents with a positive opinions, and three platforms emerge on the negative side. Consequently, their activity levels are around 1/4 and 1/6 respectively. However, notice that in this setting, the probability of silence is also clearly above zero. 

\begin{figure}[ht]
	\centering
	\includegraphics[width=0.99\linewidth]{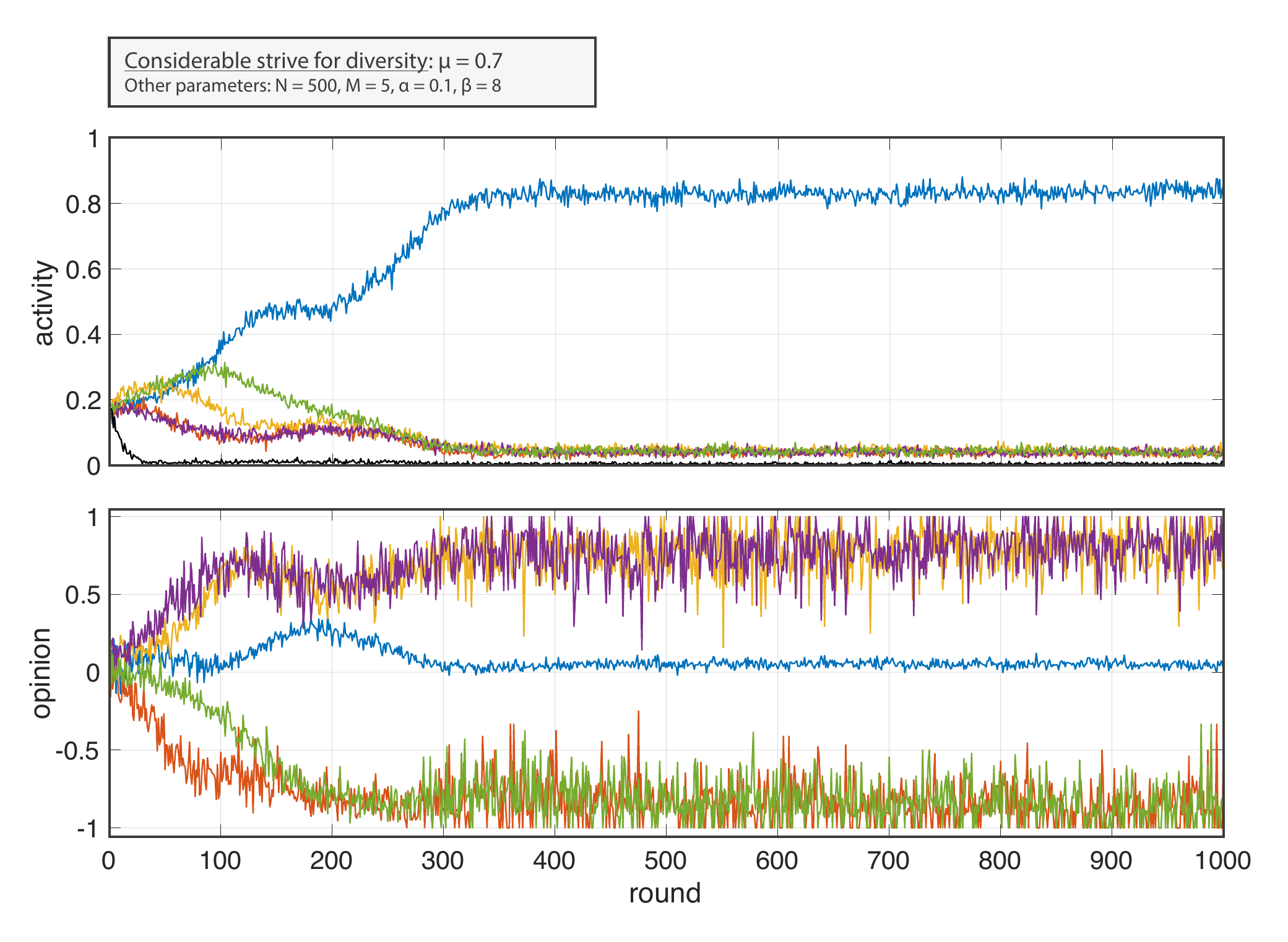}
	\caption{Second run for $\mu = 0.7$, showing the emergence of a single diverse platform (mega-platform) while others become marginalized (referred to as Run B). Other parameters: $N = 500, M = 5, \alpha = 0.1, \beta = 8$.}
	\label{fig:M07mega}
\end{figure}

The second run (run B) shows a qualitatively different trend. Here, a single platform (blue) provides a space for interaction for agents with opposing opinions. After a short period in which the platform becomes more populated by agents with a positive opinion (rounds 100 to 300), the platform stabilizes as a space on which agents from both sides express opinions. 
With an activity of more than 80\%, the blue platform grows very large and we shall refer to this as a mega-platform. 
The remaining platforms are all clearly on one side of the opinion scale connecting agents with equal opinions. But they remain marginal, as their activity converges to a low level.

\paragraph{Complex transient dynamics.}
The model exhibits rather complex transient phenomena which is shown by a realization with a slightly increased $\mu = 0.75$ in Fig. \ref{fig:N500M5B8M075}.
As before, this run features the emergence of a mega-platform on which around 80\% of activity takes place.
However, the process leading to this final stationary situation is highly non-trivial. First, two highly active platform featuring diverse opinions emerge (green and violet in this case). One platform (orange) stabilizes on the negative and the remaining two (blue and yellow) on the positive side. Then both most active platforms tend to opposing opinion sides affording less diversity. Only the green platform survives this competition and stabilizes at a high level of over 80 \% of activity.

\begin{figure}[ht]
	\centering
	\includegraphics[width=0.99\linewidth]{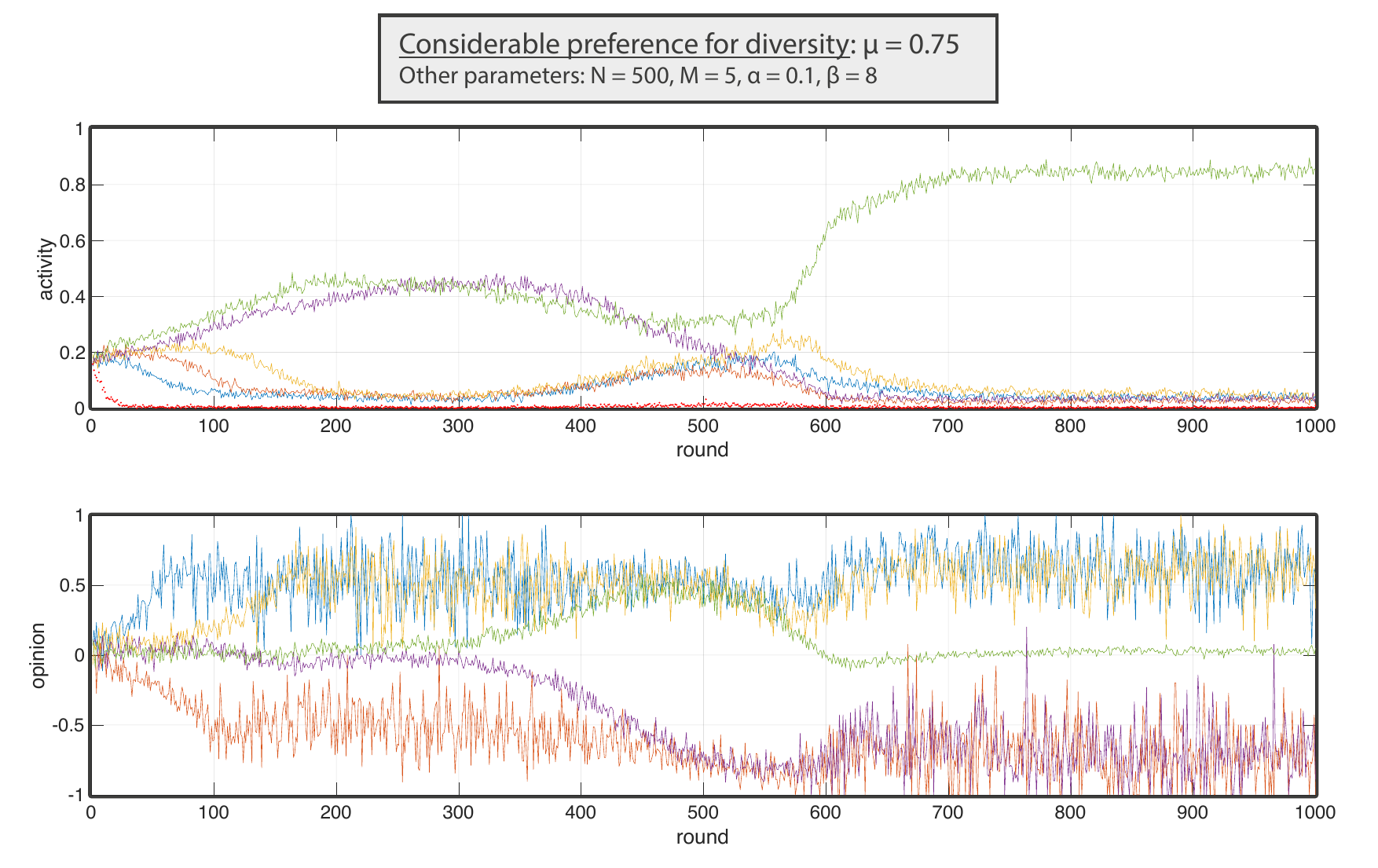}
	\caption{An example run which highlights complex transitions, with diversity initially emerging before one dominant platform forms (\( \mu = 0.75 \)). Other parameters: $N = 500, M = 5, \alpha = 0.1, \beta = 8$.}
	\label{fig:N500M5B8M075}
\end{figure}

\subsection{Platform rewards}
\label{sec:platformrewards}

In this section, we return to the two model realizations with $\mu = 0.7$ and compare the two qualitatively different outcomes with respect to how rewarding the emergent platform market is for users.
We base our definitions on the actual user experiences in the simulation, tracking the rewards agents obtained during in each round (Eq. \ref{eq:platformreward}).
For each single platform, the upper two panels in Fig. \ref{fig:M07rewards} show the average reward agents experience when active on the respective platform.

\begin{figure}[ht]
	\centering
	\includegraphics[width=0.99\linewidth]{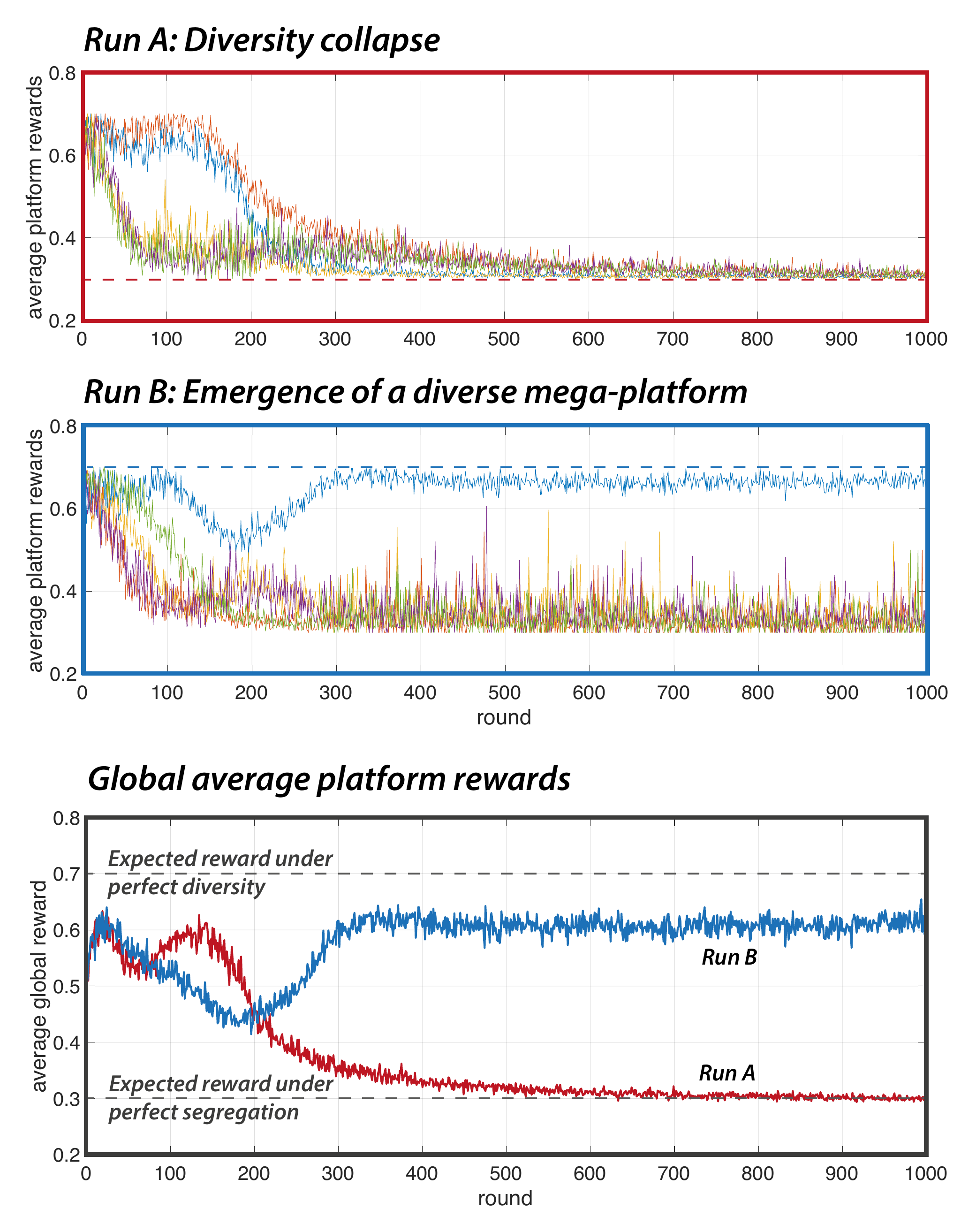}
	\caption{Comparing the two runs with $\mu = 0.7$ with respect to how rewarding the emergent platform market is for users. The dashed lines mark the expected rewards given $\mu = 0.7$ under the extreme (equilibrium) conditions of perfect diversity (global optimum) and complete segregation (suboptimal equilibrium).}
	\label{fig:M07rewards}
\end{figure}

Run A shows the system's tendency to segregate into polarized echo chambers, where agents only interact with like-minded peers and receive low rewards.
Two initially diverse platforms collapse after an initial period, user rewards approach a value of $r_k = 0.3 = (1-\mu)$ on all platforms. This corresponds to the reward that would be expected in a segregated and polarized platform market. Agents only encounter like-minded others and draw a reward of $(1-\mu)$ from these interactions. 

In contrast, Run B illustrates the emergence of a single dominant, diverse platform, where agents benefit from cross-opinion interactions, resulting in higher rewards.
On the single mega-platform, the experienced reward is close to $r_k \approx 0.7 = \mu$. This value corresponds to the theoretically expected reward if all agents meet in perfectly balanced platforms.
The remaining four platforms are marginal, but strongly opinionated so that opinion expression on them features rewards $r_k \approx 0.3$ of social approval.

This user-centric perspective on the degree to which emergent platform markets serve the needs of agents allows for a qualification of different platform constellations into optimal and suboptimal equilibria. 
Clearly, if $\mu = 0$ all agents strongly prefer interaction with agents of the same opinion. Hence, agents are maximally satisfied if they meet in homogeneous echo chambers that emerge under this parameter. More precisely, if $0 \leq \mu \leq 0.5$ in Eq. \ref{eq:rewards}, agents experience a polarized platform market maximally positive and can expect a reward $E^*_{\mu \leq 0.5}(r) = (1-\mu)$ from engaging in opinion expression.
For values $\mu > 0.5$, on the contrary, the optimal outcome would be a platform market of perfect diversity. The maximally expected reward if all agents meet in perfectly balanced online spaces is $E^*_{\mu > 0.5}(r) = \mu$.
Perfect diversity would be the optimal solution under the parameter value which lead to run A and B.
And, as is clearly visible in the bottom panel of Fig. \ref{fig:M07rewards}, when $\mu = 0.7$ agents are on average significantly more satisfied by a market with a single mega-platform that enables inter-opinion exchanges (run B) compared to a segregated situation (run A). 

\section{Analytical characterization}
\label{sec:odeformulation}

The previous section has shown that the platform choice model exhibits complex dynamics and features a set of different final equilibrium states at different parameters values $\mu$. In this section, we aim at a qualitative understanding of the global behavior of the model based on dynamical systems theory.

\subsection{Model formulation}
\label{subsec:odederivation}

We first describe how to analytically formulate the model dynamics as a system of differential equations. For this purpose, we study the model with binary opinions ($o_i \in \{-1,1\}$), as in Section \ref{sec:binarymodel}, but restrict to two competing platforms ($M = 2$). Silence is not allowed reducing the number of available options (possible actions) to two. We analyze the dynamics at a group level. That is, we assume homogeneity over the group of supporters ($\forall i$ with $o_i = 1$) and the opposing opinion group ($\forall i$ with $o=-1$).

\paragraph{4D System.}
The system dynamics is then described by four Q-values $Q(k | o )$ with platforms $k \in \{1,2\}$ and opinions $o \in \{-1,1\}$. $Q(1|1)$, for instance, denotes the value agents with $o=1$ associate with platform 1. For further convenience we shall denote this as $Q^+_1 = Q(1|1)$. Likewise, $Q(1|-1) = Q^-_1$ is the Q-value the other opinion group with $o = -1$ assigns to platform 1. Correspondingly, we write for platform two $Q^+_2$ and $Q^-_2$. 

In each round, the opinion groups choose one of the platforms by a soft max that compares the associated Q-values ($Q^+_1$ and $Q^+_2$ for agents with opinion +1). That is,
\begin{equation}
    \begin{aligned}
    p^+_1 & = \frac{1}{1 + e^{\beta  (Q^+_2-Q^+_1)}} \\      
    p^+_2 & = \frac{1}{1 + e^{\beta  (Q^+_1-Q^+_2)}}
    \end{aligned}
    \label{eq:B:selectionprobsplus}
\end{equation}
for the group with $o = +1$, and equivalently for group $o = -1$ by replacing the superscript $+$ with $-$.

The two platform opinions $O_1$ and $O_2$ can be computed on that basis by
\begin{equation}
    O_k = \frac{p^+_k - p^-_k}{p^+_k + p^-_k},
    \label{eq:B:platformopinion}
\end{equation}
where $p^+_k - p^-_k$ captures the relative prevalence of one opinion over the other and $p^+_k + p^-_k$ is the platform activity $A_k$.
The dynamics of the system are then defined by a 4D system of differential equations
\begin{equation}
    \begin{aligned}
    \dot{Q}^+_1 & = \alpha p^+_1 \Big( r^+_1 - Q^+_1 \Big)    \\
    \dot{Q}^+_2 & = \alpha p^+_2 \Big( r^+_2 - Q^+_2 \Big)      \\
    \dot{Q}^-_1 & = \alpha p^-_1 \Big( r^-_1 - Q^-_1 \Big)    \\
    \dot{Q}^-_2 & = \alpha p^-_2 \Big( r^-_2 - Q^-_2 \Big),      
    \end{aligned}
    \label{eq:B:4Dsystem}
\end{equation}
where $\alpha$ is the learning rate.
The rewards are given as in Section \ref{sec:binarymodel} by
\begin{eqnarray}
    r^+_k = (1-\mu) O_k + \mu (1 - | O_k |) \\
    r^-_k = - (1-\mu) O_k + \mu (1 - | O_k |)
    \label{eq:B:rewardsDE}
\end{eqnarray}

\paragraph{2D System.}
We can further reduce the system by looking at the differences of Q-values for each opinion group. For this purpose, we define
\begin{equation}
    \begin{aligned}
    \Delta Q^+ & = Q^+_2 - Q^+_1   \\
    \Delta Q^- & = Q^-_2 - Q^-_1
    \end{aligned}
\end{equation}
capturing to what extent both opinion groups respectively prefer platform 2 over platform 1. We exploit the fact that the selection probabilities \eqref{eq:actionselection} depend only on $\Delta Q$.
With that, the platform opinion $O_k$ \eqref{eq:B:platformopinion} and the expected rewards $r^+_k, r^-_k$ \eqref{eq:B:rewardsDE} can be computed as before.

In order to close the dynamical equations in terms of $\Delta Q^+$ and $\Delta Q^-$, we have to make an assumption which differs from the 4D case and the ABM. Namely, we skip the frequency dependence on the selection probabilities $p^+_1,p^+_2,p^-_1,p^-_2$ in \eqref{eq:B:4Dsystem}. This implicitly assumes that agents see both rewards independent of their actual choice.\footnote{This assumption does not have a qualitative impact on the behavior of the DE system and fits with an ABM of the two-platform setting. Moreover, if we include into the ABM that agents see the reward they would have obtained in the non-chosen environment (counter-factual learning), we obtain a quantitative match between analytical formulation and the simulation model.}
With this, we obtain
\begin{equation}
    \begin{aligned}
    \dot{Q}^+_2 - \dot{Q}^+_1 & = \dot{\Delta Q}^+ = \alpha \Big( r^+_2 - r^+_1 - \Delta Q^+ \Big)    \\
    \dot{Q}^-_2 - \dot{Q}^-_1 & = \dot{\Delta Q}^- = \alpha \Big( r^-_2 - r^-_1 - \Delta Q^- \Big).
    \end{aligned}
    \label{eq:B:2Dsystem}
\end{equation}

The formulation of the model in form of a system of coupled differential equations provides a rather complete understanding of the model dynamics including the equilibrium states (fixed points) and the impact of initial conditions. A global understanding of the impact of $\mu$ can be obtained with bifurcation analysis, recovering the complex structure of critical transitions observed in the simulations. 

\subsection{Phase portraits and ABM comparison}

We will first look at the phase portrait (or phase plot) of the 2D system for a given $\mu$. A phase portrait shows for any combination $(\Delta Q^+,\Delta Q^-)$ in which direction the dynamical process will evolve under the model rules. 
Let us briefly recall how to interpret the two evolving variables $(\Delta Q^+,\Delta Q^-)$. If $\Delta Q^+$ is positive, this means that positive opinion group ($o = 1$) favors platform 2 over 1, whereas a negative value indicates that platform 1 is favored. Likewise, the other opinion group ($o = -1$) prefers platform 2 if $\Delta Q^- > 0$. If $\Delta Q^+ < 0$ and $\Delta Q^- < 0$, for instance, both groups prefer posting on platform 1. The dynamical system \eqref{eq:B:2Dsystem} describes how these relative preferences evolve in time.
In Fig. \ref{fig:abm_phaseplot}, we show these temporal changes as a vector plot for $\mu = 0.7$, where simulations revealed the co-existence of two kinds of equilibria.

Note that phase plots show how the system would evolve for any specific initial condition. In a numerical simulation model, generating such a global picture would require a large number of repeated simulations with systematically varied initial conditions. Moreover, phase plots enable a graphical solution with respect to the different equilibrium states to which the model may settle. To identify these equilibria, we visualize in Fig. \ref{fig:abm_phaseplot} the null-clines with $\dot{\Delta Q}^+ = 0$ and $\dot{\Delta Q}^- = 0$ superimposed on the phase plot. Their intersections indicate the fixed points of the system, representing configurations in which no further change is expected. These fixed points can be either stable or unstable; stability is assessed through the Jacobian matrix, as detailed in Appendix \ref{sec:appendixa}.

\begin{figure}[ht]
	\centering
	\includegraphics[width=0.99\linewidth]{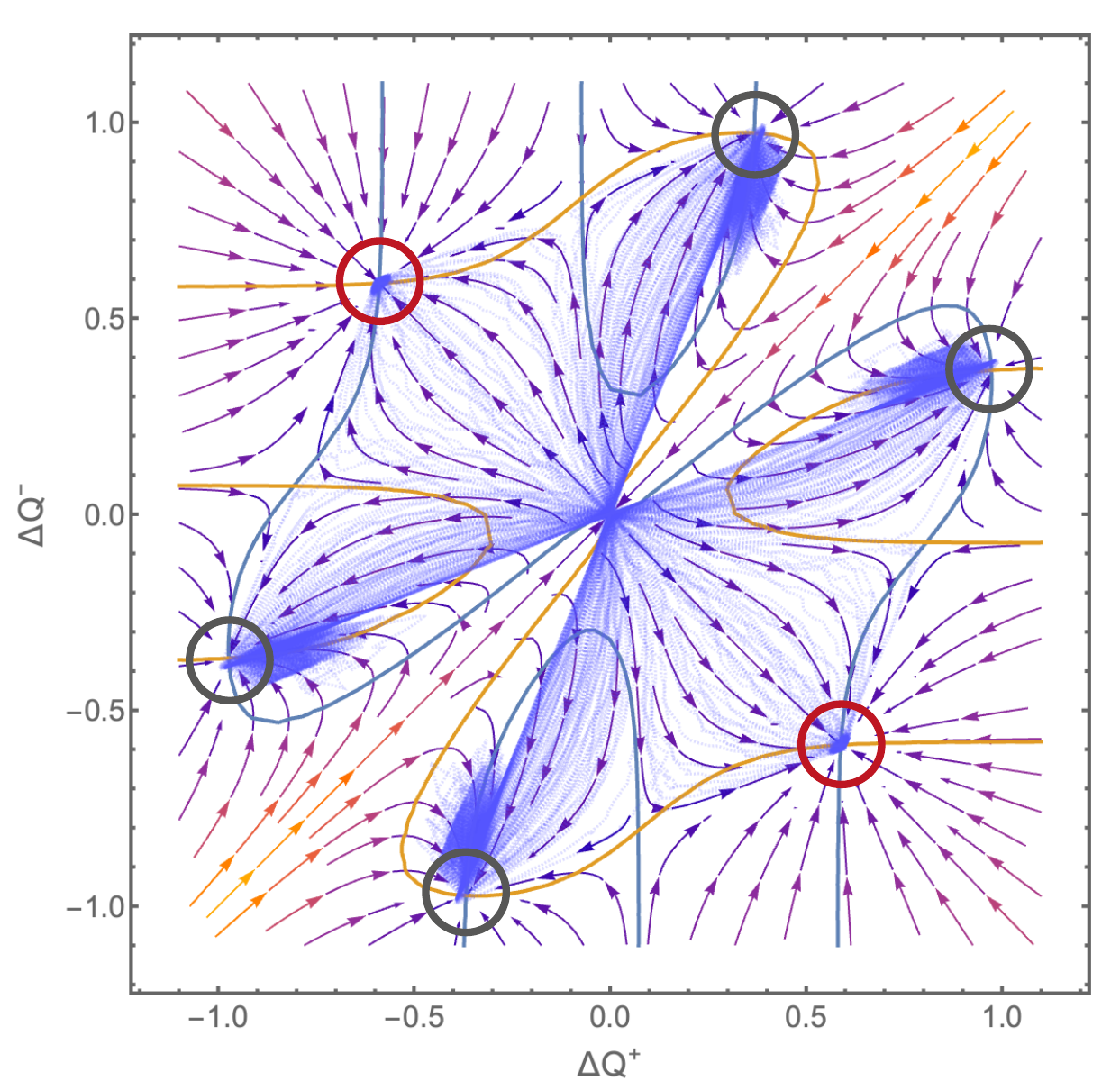}
	\caption{Phase portrait of the analytical model and comparison with 500 trajectories of an aligned ABM for $\mu = 0.7$ and $\beta = 8$. The blue and yellow curves show the null-clines of the two differential equations \eqref{eq:B:2Dsystem}. The fixed points correspond to the intersections of these curves and the stable ones are marked by the circles (red for platform segregation, gray for a single dominant platform). In blue, $500$ independently generated ABM trajectories are super-imposed. All model realizations are initialized at $(\Delta Q^+,\Delta Q^-)=(0, 0)$ and run for $1000$ rounds. Note that all runs terminate near stable fixed points and the six different equilibria are visited with non-zero probability.}
	\label{fig:abm_phaseplot}
\end{figure}

This analysis accurately recovers the complex structure of equilibria observed for $\mu = 0.7$ in the simulation section. There are six stable points for the analytical two-platform model. Four of them correspond to the mega-platform constellation where both opinion groups strongly prefer the same platform, marginalizing others to extremity. In our case of two platforms, the secondary platform is also strongly opinionated, because only one group interacts on it at a significant rate, while the other group avoids it altogether. This is reflected in the symmetry of the fixed points with respect to the diagonal. The remaining two stable fixed points, marked red, correspond to equilibria where one group prefers platform 1 and the other group platform 2 (i.e. $\Delta Q^+ < 0$ and $\Delta Q^- > 0$ or $\Delta Q^- < 0$ and $\Delta Q^+ > 0$). 

To provide further justification that the analytical formulation of the models captures ABM dynamics we also show 500 realizations of an ABM with two platforms. The model has been aligned, omitting silence and including counter-factual learning as assumed when reducing from the 4D to the 2D version of the analytical model. All model runs are started with all Q-values at zero so that the two platforms are chosen with equal probability in the beginning. To project the simulations into the 2D phase plot, we measure the mean $\Delta Q_i^+$ and $\Delta Q_i^-$ over the two opinion groups. We observed that all trajectories quickly leave this initial point at the center $(\Delta Q^+,\Delta Q^-)=(0, 0)$ and approach one of the stable fixed points. Interestingly, due to the stochasticity of the ABM (exploration) all equilibria are met by some realizations. 

\subsection{The global picture}

\begin{figure*}[ht]
	\centering
	\includegraphics[width=0.99\linewidth]{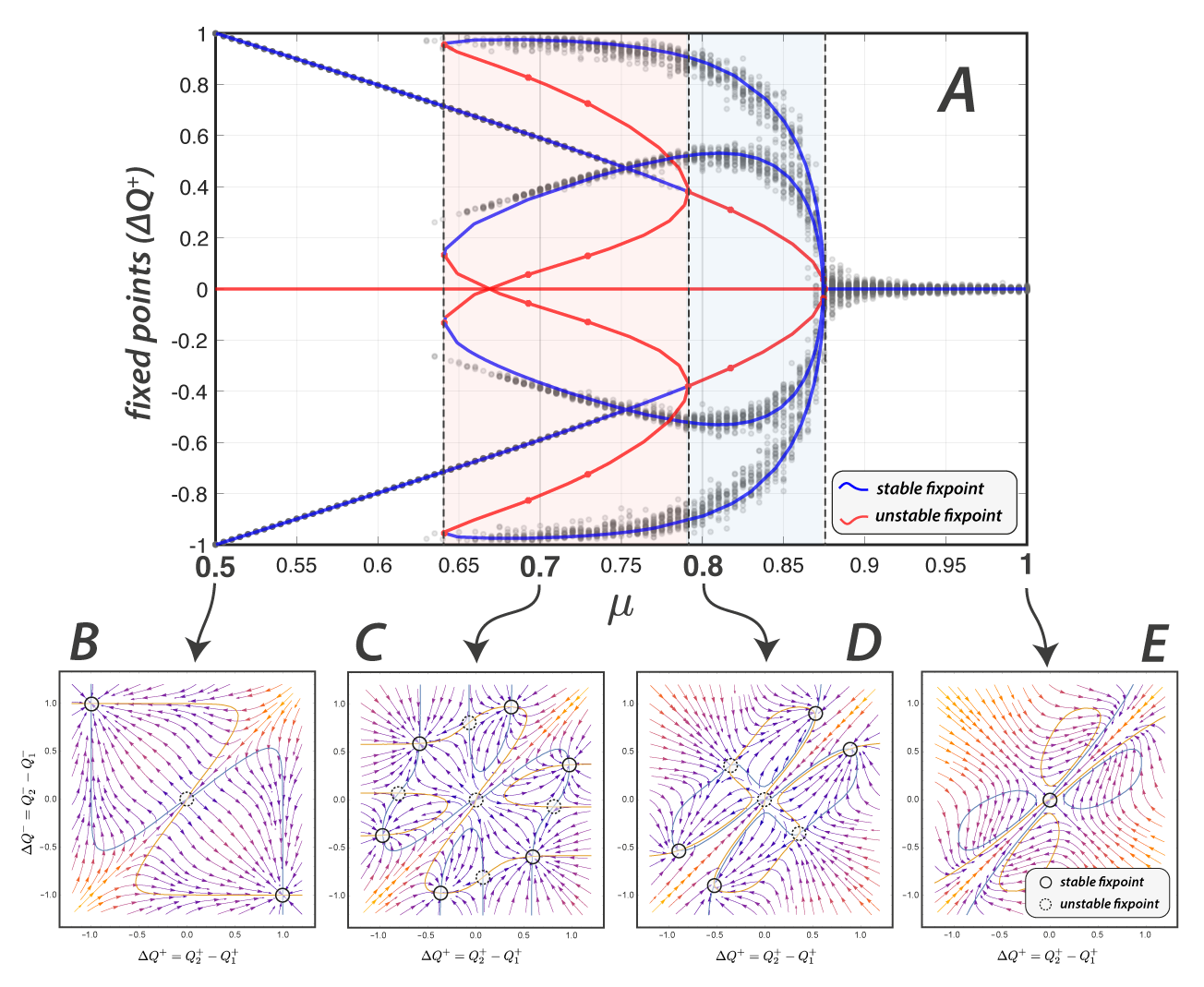}
	\caption{
 Global overview of the model behavior. 
 Panel A: 
 Bifurcation diagram tracing the fixed points as $\mu$ increases from 0.5 to 1. Blue corresponds to stable fixed points, red corresponds to unstable ones. Data from ABM simulations is show by the gray dots. Each point shows the $\Delta Q$'s in the stable state of a single model run (after $T = 1000$ steps, temporal average over the last 100 steps). For each $\mu \in [0.5,0.505,\ldots,1.0]$, 100 model runs with $N = 1000$ agents are shown
 ($10100$ realizations). The matching between simulations and theory is remarkable.
 Panel B,C,D and E: Phase portrait of the system for four different values of $\mu$. The blue and yellow curves show the null-clines of the two differential equations \eqref{eq:B:2Dsystem}. The fixed points correspond to the intersections of these curves and are marked by the black circles (solid for stable, dashed for unstable fixed points).
  }
	\label{fig:globalsystemanalysis}
\end{figure*}

While phase plots shows how the platform choice model behaves for a given $\mu$, bifurcation analysis provides a global overview of the system transitions as the balance between social approval and diversity changes. Bifurcations mark the critical parameter values $\mu$ at which the global system dynamics undergoes qualitative changes. System transitions -- sometimes called tipping points -- are characterized by the emergence of new fixed points, a change of their stability or vanishing ones. We leverage the MatCont software package \citep{dhooge2008new} which works with a symbolic representing the ODE system \eqref{eq:B:2Dsystem} and run an associated ODE solver.
The result is shown in the main Panel A of Fig. \ref{fig:globalsystemanalysis} by the blue and red curves, corresponding to stable and unstable equilibria respectively.\footnote{Notice that we only visualize $\Delta Q^+$ in the bifurcation plot because the system is symmetric and the plot for $\Delta Q^-$ would be equivalent. See Appendix \ref{sec:appendixb}.}

We again compare this analysis with a computational experiment performing systematic simulations of the aligned two-platform ABM. For each $\mu \in [0.5,0.505,\ldots,1.0]$ (101 sample points), we run 100 simulations with $N = 1000$ agents. We measure the emergent stable state after $T = 1000$ rounds in terms of the mean $\Delta Q_i^+$ and $\Delta Q_i^-$ defined over the two opinion groups during the last 100 rounds. The final $\Delta Q$'s for all 10100 model realizations are shown by the gray dots. We observe a remarkable matching between the ABM and the results of the bifurcation analysis. 

Beyond the excellent fit with the aligned ABM, we find that the analytical model captures the series of non-trivial phenomena observed in the simulation of the full model as the alignment-diversity balance increases towards more diversity-seeking behavior. For this purpose, a series of phase portraits for specific values of $\mu$ are shown in Panels B to E of Fig. \ref{fig:globalsystemanalysis}. These values are chosen as representatives of qualitatively different model regimes as revealed by the bifurcation analysis.

\paragraph{Suboptimal platform segregation.}
In Panel B, the phase portrait for $\mu = 0.5$ is shown. There are three fixed points. The first one at $(0,0)$ is unstable so that the system is driven away from it towards one of the two stable fixed points at around $(-1,1)$ and $(1,-1)$. In these states, the $+$ group strongly prefers one platform whereas the $-$ group prefers the other. In $(-1,1)$, for instance, $Q^+_1$ is large compared to $Q^+_2$ so that the $+$ group interacts on platform 1 with a probability close to one. Conversely, $\Delta Q^- = 1$ indicates a strong preference of the $-$ group for platform 2. We are hence in a situation of segregation where platforms polarize into strongly opinionated echo chambers.

As shown in Section \ref{sec:platformrewards}, this situation (equilibrium) is actually suboptimal as soon as $\mu > 0.5$. However, up to values of $\mu \approx 0.64$ these two equilibria are the only stable ones. This shows that the platform choice model will evolve into a suboptimal state of a segregated platform market despite the fact that all agents individually would be more satisfied by interaction on diverse online spaces in which they encounter opposing views.

\paragraph{The emergence of a second non-trivial attractor.}
As $\mu$ crosses a critical value of 0.64 (here exemplified by $\mu = 0.7$ in Panel C), the system bifurcates into a different qualitative regime and a series of new equilibria appears.  There are four additional fixed points in which both opinion groups prefer the same platform to a slightly different degree. This nicely corresponds to the emergence of a mega-platform, on which agents with both opinions meet, as shown in the simulations section. There are four different such states in the two-platform version of the model. First, platform 1 or 2 may end up being chosen by both groups. Second, the remaining marginal platform may become populated predominantly by either the $+$ or the $-$ group. Hence, as in the simulations, this second platform is far less active, but strongly opinionated. The emergence of a stable state in which one diverse platform dominates the others is one of the most surprising features of the platform choice model.

\paragraph{Coexistence of multiple equilibria.}
Still, the fixed points associated to platform polarization remain stable up to a value of $\mu = 0.875$. In this parameter regime, highlighted in Panel A by a light red shade, platform segregation and the emergence of a mega-platform coexist. We have observed this multi-stability in the aligned two-platform ABM as well as in the full model with several platforms. In this case (i.e. for $0.64 < \mu < 0.79$), the evolution of the model will strongly depend on initial conditions. Moreover, we might observe hysteresis: once the system enters a suboptimal attractor (platform segregation), it may be very difficult to return. On the other hand, the coexistence of these two equilibria provides ground for non-structural interventions such as promoting one platform as an space open to all groups, aiming to raise expectations ($Q$ values) for one platform.

\paragraph{Dominance as the only equilibrium.}
As $\mu$ crosses a value of 0.79 (exemplified by $\mu = 0.8$ in Panel D), the stable fixed points associated to platform segregation become unstable. The emergence of a single diverse mega-platform is the only stable outcome in this parameter regime ($0.79 < \mu < 0.875$). 

As far as we know, this state is unique to our model addressing segregation in a dynamic online environment. The  precise mechanism that leads to the survival of only a single diverse platform deserves further analysis.

\paragraph{Indifference.}
Finally, for $\mu > 0.875$ the fixed points associated with the mega-platform profile disappear and collapse into the fixed point at the center.  $(0,0)$ becomes the only stable solution. This is exemplified in Panel E of Fig. \ref{fig:globalsystemanalysis} for $\mu = 1$ which had also been considered in the simulation section as a limiting case. $Q$-values converge to the same values in the long run, meaning that agents are indifferent with respect to which platform they prefer. All platforms will be chosen with equal probability.

\section{Discussion}

Our study builds upon Social Feedback Theory (SFT) which has been proposed as a framework for modeling how user behavior in online environments is shaped by social reinforcement \citep{Banisch2019opinion, Banisch2022modeling}.
In repeated communication games agents express their opinions and receive feedback such as likes and downvotes from peers. 
While previous models developed within this paradigm \citep{Banisch2019opinion,jacob2023polarization,konovalova2023social,lefebvre2024roots} have focused on how users adapt their opinions based on the social feedback they receive -- suggesting that social reinforcement may be the primary mechanism responsible for increasing polarization \citep{lefebvre2024roots} and extreme opinion expression \citep{konovalova2023social} -- this work takes a different perspective. It focuses on the process of online community selection; an orthogonal but understudied facet of opinion dynamics which is especially relevant in the digital age. 

In contrast to traditional opinion dynamics models, which primarily focus on how opinions evolve through direct interactions within fixed networks, our approach emphasizes the role of platform choice as a dynamic and strategic decision made by users. For instance, models like those discussed by \cite{friedkin2011social}, \cite{flache2017models} and \cite{lorenz2021individual} typically assume that opinions change as users interact with others in predefined social networks, where exposure to like-minded or opposing views directly influences opinion shifts. 
These models often focus on opinion convergence or polarization within a single platform or network, without considering the broader context of multiple competing platforms. Our model departs from these frameworks by shifting the focus from opinion change to the choice of where to express opinions. 
This distinction is particularly relevant in an online context, where users have the freedom to navigate between platforms and adapt their behavior and online consumption based on the feedback they receive. By modeling how users with given opinions actively decide which platform to use based on rewards such as social approval or exposure to diverse views, we capture a key dynamic that is missing in many conventional opinion dynamics models addressing online phenomena \citep{maes2015will,keijzer2018communication,keijzer2022complex,kozitsin2022formal}. This shift allows us to explore the interplay between user behavior and platform structures, offering new insights into how polarization and diversity can emerge across multiple platforms in an online ecosystem, rather than within a single, closed environment.




Our findings resonate with Schelling’s segregation model, which demonstrated how individual preferences can lead to unintended collective outcomes, such as spatial segregation \citep{schelling1969models,schelling1971dynamic}. In Schelling’s model, even mild preferences for homophily resulted in sharp social divisions. Similarly, our model shows that even when users actively seek diverse interactions, online environments can still evolve into polarized, echo-chamber-like spaces. However, the model also reveals an unexpected equilibrium where users become active on a single large platform, marginalizing other platforms into extreme opinion clusters.
This state of equilibrium reflects how, under certain parameter values (such as a moderate $\mu$), users gravitate toward a platform that offers a sufficient mix of affirmation and diversity, potentially creating a "mega-platform" where the main discourse occurs, while alternative platforms become polarized and isolated.
This co-existence of multiple qualitatuively different equilibria, not found in earlier models, highlights the importance of the co-evolutionary dynamics between individual user behavior and macroscopic platform composition, adding new dimensions to the study of digital polarization.
The presence of these equilibria also suggests a strong path dependence and potential for hysteresis, where initial conditions and past states significantly shape the system's long-term outcomes.


While our model offers valuable insights into the dynamics of user behavior and platform choice, several limitations must be acknowledged. First, the model simplifies the complexity of user preferences by focusing primarily on two factors: social approval and exposure to diversity. In reality, users may be influenced by a broader range of factors, such as economic incentives, platform-specific features (e.g., ease of use, privacy concerns), or algorithmic transparency \citep{lorenzspreen2024real}. Studies have shown that users' decisions to switch platforms may also be driven by factors like community norms or platform governance \citep{gillespie2018custodians}. Incorporating a richer set of user motivations and platform characteristics could lead to more nuanced insights into the co-evolution of user behavior and platform ecosystems.

Furthermore, our model focuses on a homogeneous population of users divided along a single dimension (e.g., political opinion). In real-world online environments, users hold a diversity of opinions across multiple issues, which interact in complex ways \citep{converse2006nature,baldassarri2014neither,dellaposta2020pluralistic}. By modeling users as holding just one opinion or identity, we may overlook the potential for cross-cutting interactions across different issues, which can either exacerbate or mitigate polarization \cite{banisch2021argument,baumann2021emergence,camargo2020new}. 
The computational framework is flexible enough to expand the model to account for users’ multi-dimensional identities -- including political, cultural, and social preferences \citep{baldassarri2007dynamics}. In particular, it allows to directly incorporate heterogeneous user populations with opinions measured in surveys or on online social networking data \citep{gaumont2018reconstruction,ramaciotti2022inferring,peralta2024multidimensional}.
While complicating the computational analysis, designing empirical scenarios that account for more complexity at the level of users would allow for a more comprehensive understanding of how users engage with different platforms.

Finally, while our model focuses on competition between multiple platforms, it does not fully capture the influence of platform algorithms beyond simple reward mechanisms. Real-world platforms employ complex algorithms that curate content in highly personalized ways \citep{bakshy2015exposure}. Incorporating more detailed models of algorithmic curation -- such as a user feed -- would allow for a deeper exploration of the role these algorithms play in shaping opinion dynamics and platform choice. This would also provide more actionable insights for platform designers seeking to mitigate polarization while maintaining user engagement \citep{lazer2018science,bhadani2022political}.

Nevertheless, the findings of our model have implications for both the design and management of digital platforms, as well as for broader societal concerns about online polarization and discourse. One of the key insights from our study is the coexistence of multiple equilibrium outcomes: platforms can evolve into polarized, echo-chamber-like environments or, alternatively, a single dominant platform can emerge where opposing opinions coexist, while smaller platforms become marginalized. This highlights the critical role of platform design and algorithmic moderation in shaping these outcomes.
In practice, algorithms that prioritize user engagement -- often by amplifying content that resonates strongly with existing user preferences -- can contribute to the formation of polarized spaces. When users seek out platforms that align closely with their beliefs, social media ecosystems can naturally segregate into clusters of like-minded individuals, reinforcing the very dynamics of polarization that have been widely criticized \citep{bakshy2015exposure,sunstein2018republic}. This suggests that algorithmic personalization, though often deployed to enhance user satisfaction, may inadvertently fuel online segregation by drawing users into homophilous communities.

However, our model also points to opportunities for intervention. By adjusting platform features or modifying recommendation algorithms to promote exposure to diverse viewpoints, it may be possible to foster environments where a dominant, more inclusive platform prevails. For example, algorithms that intentionally balance between reinforcing social approval and offering diverse content could create spaces where users are more likely to encounter opposing views, potentially reducing polarization \citep{lorenzspreen2021algorithmic,willaert2021opinion}. However, while our model shows that exposure to diverse opinions can help reduce polarization within a dominant platform, it does not necessarily follow that smaller platforms can be prevented from becoming ideologically extreme through content moderation policies alone. In fact, smaller platforms may attract niche user groups precisely because of their more focused, often extreme, content, and the marginalization of these platforms into more radical spaces could be an inevitable consequence of market dynamics rather than content policies. Efforts to reduce the visibility of extreme content on larger platforms, for instance, might push certain user groups toward alternative platforms, further entrenching ideological divides online \citep{cinelli2021echo}.

This highlights a significant challenge: while promoting content diversity and algorithmic balance on larger platforms could foster more inclusive spaces, it is unlikely that such efforts would fully prevent the concentration of extreme views on smaller platforms. As users seek spaces that align more closely with their ideological preferences, particularly when faced with moderation on larger platforms, smaller, less moderated platforms may become hubs for more polarized and extreme discourse. This suggests that platform designers and policymakers need to carefully consider the trade-offs involved in content moderation and algorithmic design, especially as they relate to the broader platform ecosystem.

\section{Conclusion}

In this paper, we introduced a dynamic model of platform choice, extending Social Feedback Theory to capture how users’ decisions to engage with particular platforms contribute to the broader online ecosystem.
Through a rigorous mathematical analysis, including bifurcation analysis, and agent-based simulations, we uncovered two key equilibrium outcomes: one where users polarize into echo chambers on separate platforms, and another where a dominant platform fosters diversity by accommodating opposing opinions, marginalizing all other platforms to extremity. Our model offers new insights into how user preferences for diversity or social approval drive the evolution of online environments.


This approach provides a new perspective on online polarization by focusing on the adaptive choices users with different opinions make across competing platforms, rather than restricting dynamics to a single platform environment. 
It highlights the complex interplay between adaptive user preferences and emergent platform structure, demonstrating that even when users value diversity, the feedback loops between platform selection and social feedback can still lead to segregation.
In doing so, our model sheds light on the digital public sphere, illustrating how user motivations for either social reinforcement or diverse interactions can shape the kinds of discourse that unfold across platforms. Our findings underscore the co-evolutionary relationship between individual behavior and platform composition, revealing how user motivations, path dependencies, and initial conditions collectively influence the structure and inclusivity of online public spaces.

Overall, this study offers both theoretical and practical insights into the design of social media platforms. It suggests that platform choices driven by user preferences play a significant role in structuring digital public spaces, with implications for mitigating polarization. Future research will build on this framework by integrating more realistic user identities and algorithmic factors, contributing to the design of robust strategies for fostering a sustainable digital ecosystem.

\subsection*{Acknowledgements}

This project has received funding from the European Union’s Horizon Europe programme under grant agreement No 101094752 (SoMe4Dem).

\bibliographystyle{apalike}
\bibliography{references.bib}

\clearpage
\onecolumn
\appendix

\section{Derivation of Jacobian}
\label{sec:appendixa}
In this section, we analytically derive the Jacobian matrix $J$ associated with the $2D$ ODE system in Equation \ref{eq:B:2Dsystem}. In general, fixed points in a dynamical system can be either \emph{stable} or \emph{unstable}. The former implies that dynamics within a local neighborhood will exponentially converge to the fixed point. On the other hand, the latter corresponds to a lack of convergence. The Jacobian represents a linear approximation of a dynamical system and provides an understanding of underlying stability. Specifically, we can check the stability of a candidate fixed point by substituting its coordinate values into the matrix $J$ and then checking its eigenvalues. If all eigenvalues are negative, then we have a stable fixed point; otherwise, we have an unstable fixed point \citep{chasnov_nonlinear_2019}. 

\noindent
\newline
To derive the Jacobian, we first define the functions $f(\Delta Q^+, \Delta Q^-) = \dot{\Delta Q^+}$ and $g(\Delta Q^+, \Delta Q^-) = \dot{\Delta Q^-}$. Then, finding the Jacobian for Equation \ref{eq:B:2Dsystem} entails finding the following set of derivatives:

\begin{equation}
\label{eq:jacobianODEsimplified} 
J = \left[\begin{array}{cc}
     \frac{\partial f}{\partial\Delta Q^+} & \frac{\partial f}{\partial\Delta Q^-} \\
    \frac{\partial g}{\partial\Delta Q^+} & \frac{\partial g}{\partial\Delta Q^-} \\
\end{array}\right]
\end{equation}

\noindent
We begin by expanding out functions $f$ and $g$. For the former, we will have:

\begin{equation}
\label{eq:expandedfunctionf} 
\begin{split}
    f(\Delta Q^+, \Delta Q^-) &= \alpha(r_2^+ - r_1^+ - \Delta Q^+) \\
    &= \alpha([(1 - \mu)O_2 + \mu(1 - |O_2|)] - [(1 - \mu) O_1 + \mu (1 - |O_1|)] - \Delta Q^+)
\end{split}
\end{equation}

\noindent
Similarly, we will have the following for $g$:

\begin{equation}
\label{eq:expandedfunctiong} 
\begin{split}
    g(\Delta Q^+, \Delta Q^-) &= \alpha(r_2^- - r_1^- - \Delta Q^-) \\
    &= \alpha([-(1 - \mu)O_2 + \mu(1 - |O_2|)] - [-(1 - \mu) O_1 + \mu (1 - |O_1|)] - \Delta Q^-)
\end{split}
\end{equation}

\noindent
We can then compute derivatives using the chain rule. 

\begin{equation}
\label{eq:jacobianderivativesf1} 
\begin{split}
    \frac{\partial f}{\partial\Delta Q^+} &= \alpha\left(\left[\left(1 - \mu\right)\frac{\partial O_2}{\partial \Delta Q^+}\right.\right. + \left.\mu\left(1 - \text{sign}\left(O_2\right) \cdot \frac{\partial O_2}{\partial \Delta Q^+}\right)\right] \\\ 
    &- \left[\left(1 - \mu\right) \frac{\partial O_1}{\partial \Delta Q^+} \right.\left. \left. + \mu \left(1 - \text{sign}\left(O_1\right) \cdot \frac{\partial O_1}{\partial \Delta Q^+}\right)\right] - 1\right) \\
\end{split}
\end{equation}

\begin{equation}
\label{eq:jacobianderivativesf2} 
\begin{split}
    \frac{\partial f}{\partial\Delta Q^-} &= \alpha\left(\left[\left(1 - \mu\right)\frac{\partial O_2}{\partial \Delta Q^-}\right.\right. + \left.\mu\left(1 - \text{sign}\left(O_2\right) \cdot \frac{\partial O_2}{\partial \Delta Q^-}\right)\right] \\
    &- \left[\left(1 - \mu\right) \frac{\partial O_1}{\partial \Delta Q^-} \right.\left. \left. + \mu \left(1 - \text{sign}\left(O_1\right) \cdot \frac{\partial O_1}{\partial \Delta Q^-}\right)\right]\right) \\
\end{split}
\end{equation}

\begin{equation}
\label{eq:jacobianderivativesg1} 
\begin{split}
    \frac{\partial g}{\partial\Delta Q^+} &= \alpha\left(\left[-\left(1 - \mu\right)\frac{\partial O_2}{\partial \Delta Q^+}\right.\right. + \left.\mu\left(1 - \text{sign}\left(O_2\right) \cdot \frac{\partial O_2}{\partial \Delta Q^+}\right)\right] \\
    &- \left[-\left(1 - \mu\right) \frac{\partial O_1}{\partial \Delta Q^+} \right. \left. \left. + \mu \left(1 - \text{sign}\left(O_1\right) \cdot \frac{\partial O_1}{\partial \Delta Q^+}\right)\right]\right) \\
\end{split}
\end{equation}

\begin{equation}
\label{eq:jacobianderivativesg2} 
\begin{split}
    \frac{\partial g}{\partial\Delta Q^-} &= \alpha\left(\left[-\left(1 - \mu\right)\frac{\partial O_2}{\partial \Delta Q^-}\right.\right. + \left.\mu\left(1 - \text{sign}\left(O_2\right) \cdot \frac{\partial O_2}{\partial \Delta Q^-}\right)\right] \\
    &- \left[-\left(1 - \mu\right) \frac{\partial O_1}{\partial \Delta Q^-} \right.\left. \left. + \mu \left(1 - \text{sign}\left(O_1\right) \cdot \frac{\partial O_1}{\partial \Delta Q^-}\right)\right] - 1\right) \\
\end{split}
\end{equation}

\noindent
\newline
Here, the function $\text{sign}(x) = \begin{cases}
    1 & x \geq 0 \\
    0 & x < 0
\end{cases}$. Note that the presence of the absolute value function in $f$ and $g$ implies that the Jacobian is not differentiable at $x = 0$.

\noindent
\newline
It remains to compute the derivatives $\frac{\partial O_1}{\partial \Delta Q^+}, \frac{\partial O_1}{\partial \Delta Q^-}, \frac{\partial O_2}{\partial \Delta Q^+}, \frac{\partial O_2}{\partial \Delta Q^+}$. To do so, we first derive an expression for the \textit{opinions} $O_1$ and $O_2$ from Equation \ref{eq:B:platformopinion} in terms of the coordinate values $\Delta Q^+$ and $\Delta Q^-$. This is done by using the definition of probabilities from Equation \ref{eq:B:selectionprobsplus} then simplifying:

\begin{equation}
\label{eq:averageQvalueplatformopinionQvalues}
\begin{array}{c}
     O_1(\Delta Q^+, \Delta Q^-) = \frac{e^{\beta \Delta Q^-} - e^{\beta \Delta Q^+}}{2 + e^{\beta \Delta Q^+} + e^{\beta \Delta Q^-}}  \\
     O_2(\Delta Q^+, \Delta Q^-) = \frac{e^{-\beta \Delta Q^-} - e^{-\beta \Delta Q^+}}{2 + e^{-\beta \Delta Q^+} + e^{-\beta \Delta Q^-}}
\end{array}
\end{equation}

\noindent
We now compute the associated derivatives using the chain rule and quotient rule. For notational convenience we define the numerators from Equation \ref{eq:averageQvalueplatformopinionQvalues} as $\mathcal{N}_{O_1} := e^{\beta \Delta Q^-} - e^{\beta \Delta Q^+}$ and $\mathcal{N}_{O_2} := e^{-\beta \Delta Q^-} - e^{-\beta \Delta Q^+}$. We additionally define the denominators from Equation \ref{eq:averageQvalueplatformopinionQvalues} as $\mathcal{D}_{O_1} := 2 + e^{\beta \Delta Q^+} + e^{\beta \Delta Q^-}$ and $\mathcal{D}_{O_2} := 2 + e^{-\beta \Delta Q^+} + e^{-\beta \Delta Q^-}$.

\begin{equation}
\label{eq:jacobianderivativeso1qplus} 
\begin{split}
    \frac{\partial O_1}{\partial\Delta Q^+} &= \frac{(-\beta e^{\beta \Delta Q^+}) \cdot \mathcal{D}_{O_1} - \mathcal{N}_{O_1} \cdot (\beta e^{\beta \Delta Q^+})}{(\mathcal{D}_{O_1})^2} \\
    &= \frac{-2\beta e^{\beta \Delta Q^+} (1 + e^{\beta \Delta Q^-})}{(2 + e^{\beta \Delta Q^+} + e^{\beta \Delta Q^-})^2}
\end{split}
\end{equation}
\begin{equation}
\label{eq:jacobianderivativeso1qminus} 
\begin{split}
    \frac{\partial O_1}{\partial\Delta Q^-} &= \frac{(\beta e^{\beta \Delta Q^-}) \cdot \mathcal{D}_{O_1} - \mathcal{N}_{O_1} \cdot (\beta e^{\beta \Delta Q^-})}{(\mathcal{D}_{O_1})^2} \\
    &= \frac{2\beta e^{\beta \Delta Q^-} (1 + e^{\beta \Delta Q^+})}{(2 + e^{\beta \Delta Q^+} + e^{\beta \Delta Q^-})^2}
\end{split}
\end{equation}

\begin{equation}
\label{eq:jacobianderivativeso2qplus} 
\begin{split}
    \frac{\partial O_2}{\partial\Delta Q^+} &= \frac{(\beta e^{-\beta \Delta Q^+}) \cdot \mathcal{D}_{O_2} - \mathcal{N}_{O_2} \cdot (-\beta e^{-\beta \Delta Q^+})}{(\mathcal{D}_{O_2})^2} \\
    &= \frac{2\beta e^{-\beta \Delta Q^+} (1 + e^{-\beta \Delta Q^-})}{(2 + e^{-\beta \Delta Q^+} + e^{-\beta \Delta Q^-})^2}
\end{split}
\end{equation}
\begin{equation}
\label{eq:jacobianderivativeso2qminus} 
\begin{split}
    \frac{\partial O_2}{\partial\Delta Q^-} &= \frac{(-\beta e^{-\beta \Delta Q^-}) \cdot \mathcal{D}_{O_2} - \mathcal{N}_{O_2} \cdot (-\beta e^{-\beta \Delta Q^-})}{(\mathcal{D}_{O_2})^2} \\
    &= \frac{-2\beta e^{-\beta \Delta Q^-} (1 + e^{-\beta \Delta Q^+})}{(2 + e^{-\beta \Delta Q^+} + e^{-\beta \Delta Q^-})^2}
\end{split}
\end{equation}

\noindent
We can then plug these derivatives into Equations \ref{eq:jacobianderivativesf1}, \ref{eq:jacobianderivativesf2}, \ref{eq:jacobianderivativesg1}, and \ref{eq:jacobianderivativesg2} to get the analytically defined Jacobian in terms of the coordinate values $\Delta Q^+$ and $\Delta Q^-$.

\section{Equivalence of bifurcation plots}
\label{sec:appendixb}



We now prove that the bifurcation plot (Figure \ref{fig:globalsystemanalysis}, Panel A) is equivalent for  $\Delta Q^+$ and $\Delta Q^-$ because the system is symmetric. To start, we first demonstrate the following claims involving symmetry of fixed points in the dynamical system.

\begin{restatable}[(Fixed points are symmetric across line $y = x$)]{claim}{symmetricone}
\label{claim:symmetric1}

Suppose that a point $\Delta Q^+ = x$ and $\Delta Q^- = y$ for $x, y \in \mathbb{R}$ is a fixed point for the system in Equation \ref{eq:B:2Dsystem}. Then, the point $\Delta Q^+ = y$ and $\Delta Q^- = x$ is also a fixed point.

\end{restatable}

\begin{proof}
    We are given a fixed point $\Delta Q^+ = x$ and $\Delta Q^- = y$. Then, leveraging the notation from Equations \ref{eq:expandedfunctionf} and \ref{eq:expandedfunctiong} from Appendix \ref{sec:appendixa} we note that this implies:

    \[f(x, y) = g(x, y) = 0\]

    \noindent
    \newline
    We now evaluate $f(y, x)$ and $g(y, x)$. To do so, we first make some observations on the \textit{opinion} expressions from Equation \ref{eq:averageQvalueplatformopinionQvalues}. Specifically, we can evaluate $O_1(y, x)$ as follows:

    \eqnsplit{
        O_1(y, x) &= \frac{e^{\beta x} - e^{\beta y}}{2 + e^{\beta y} + e^{\beta x}} \\
        &= -\left(\frac{e^{\beta y} - e^{\beta x}}{2 + e^{\beta y} + e^{\beta x}}\right) \\
        &= -O_1(x, y)
    }

    \noindent
    Additionally, we will have:

    \eqnsplit{
        O_2(y, x) &= \frac{e^{-\beta x} - e^{-\beta y}}{2 + e^{-\beta y} + e^{-\beta x}} \\
        &= -\left(\frac{e^{-\beta y} - e^{-\beta x}}{2 + e^{-\beta y} + e^{-\beta x}}\right) \\
        &= -O_2(x, y)
    }

    \noindent
    We now can write out an expression for $f(y, x)$ using Equation \ref{eq:expandedfunctionf}:

    \eqnsplit{
        f(y, x) &= \alpha(r_2^+ - r_1^+ - y) \\
        &= \alpha([(1 - \mu)O_2(y, x) + \mu(1 - |O_2(y, x)|)] - [(1 - \mu) O_1(y, x) + \mu (1 - |O_1(y, x)|)] - y) \\
        &= \alpha([-(1 - \mu)O_2(x, y) + \mu(1 - |O_2(x, y)|)] - [-(1 - \mu) O_1(x, y) + \mu (1 - |O_1(x, y)|)] - y) \\
        &= g(x, y) \\
        &= 0
    }

    \noindent
    Similarly we can write out an expression for $g(y, x)$ using Equation \ref{eq:expandedfunctiong}:

    \eqnsplit{
        g(y, x) &= \alpha(r_2^- - r_1^- - x) \\
        &= \alpha([-(1 - \mu)O_2(y, x) + \mu(1 - |O_2(y, x)|)] - [-(1 - \mu) O_1(y, x) + \mu (1 - |O_1(y, x)|)] - x) \\
        &= \alpha([(1 - \mu)O_2(x, y) + \mu(1 - |O_2(x, y)|)] - [(1 - \mu) O_1(x, y) + \mu (1 - |O_1(x, y)|)] - x) \\
        &= f(x, y) \\
        &= 0
    }

    \noindent
    Thus, we have shown that $f(y, x) = g(y, x) = 0$. This implies that the coordinate pair $\Delta Q^+ = y$ and $\Delta Q^- = x$ is a fixed point.
    
\end{proof}

\begin{restatable}[(Fixed points are symmetric across line $y = -x$)]{claim}{symmetrictwo}
\label{claim:symmetric2}

Suppose that a point $\Delta Q^+ = x$ and $\Delta Q^- = y$ for $x, y \in \mathbb{R}$ is a fixed point for the system in Equation \ref{eq:B:2Dsystem}. Then, the point $\Delta Q^+ = -y$ and $\Delta Q^- = -x$ is also a fixed point.

\end{restatable}

\begin{proof}
    We are given a fixed point $\Delta Q^+ = x$ and $\Delta Q^- = y$. Then, leveraging the notation from Equations \ref{eq:expandedfunctionf} and \ref{eq:expandedfunctiong} from Appendix \ref{sec:appendixa} we note that this implies:

    \[f(x, y) = g(x, y) = 0\]

    \noindent
    \newline
    We now evaluate $f(-y, -x)$ and $g(-y, -x)$. To do so, we first make some observations on the \textit{opinion} expressions from Equation \ref{eq:averageQvalueplatformopinionQvalues}. Specifically, we can evaluate $O_1(-y, -x)$ as follows:

    \eqnsplit{
        O_1(-y, -x) &= \frac{e^{-\beta x} - e^{-\beta y}}{2 + e^{-\beta y} + e^{-\beta x}} \\
        &= -\left(\frac{e^{-\beta y} - e^{-\beta x}}{2 + e^{-\beta y} + e^{-\beta x}}\right) \\
        &= -O_2(x, y)
    }

    \noindent
    Additionally, we will have:

    \eqnsplit{
        O_2(-y, -x) &= \frac{e^{\beta x} - e^{\beta y}}{2 + e^{\beta y} + e^{\beta x}} \\
        &= -\left(\frac{e^{\beta y} - e^{\beta x}}{2 + e^{\beta y} + e^{\beta x}}\right) \\
        &= -O_1(x, y)
    }

    \noindent
    We now can write out an expression for $f(-y, -x)$ using Equation \ref{eq:expandedfunctionf}:

    \eqnsplit{
        f(-y, -x) &= \alpha(r_2^+ - r_1^+ + y) \\
        &= \alpha([(1 - \mu)O_2(-y, -x) + \mu(1 - |O_2(-y, -x)|)] - [(1 - \mu) O_1(-y, -x) + \mu (1 - |O_1(-y, -x)|)] + y) \\
        &= \alpha([-(1 - \mu)O_1(x, y) + \mu(1 - |O_1(x, y)|)] - [-(1 - \mu) O_2(x, y) + \mu (1 - |O_2(x, y)|)] + y) \\
        &= -\alpha([-(1 - \mu)O_2(x, y) + \mu(1 - |O_2(x, y)|)] - [-(1 - \mu) O_1(x, y) + \mu (1 - |O_1(x, y)|)] - y) \\
        &= -g(x, y) \\
        &= 0
    }

    \noindent
    Similarly we can write out an expression for $g(-y, -x)$ using Equation \ref{eq:expandedfunctiong}:

    \eqnsplit{
        g(-y, -x) &= \alpha(r_2^- - r_1^- + x) \\
        &= \alpha([-(1 - \mu)O_2(-y, -x) + \mu(1 - |O_2(-y, -x)|)] - [-(1 - \mu) O_1(-y, -x) + \mu (1 - |O_1(-y, -x)|)] + x) \\
        &= \alpha([(1 - \mu)O_1(x, y) + \mu(1 - |O_1(x, y)|)] - [(1 - \mu) O_2(x, y) + \mu (1 - |O_2(x, y)|)] + x) \\
        &= -\alpha([(1 - \mu)O_2(x, y) + \mu(1 - |O_2(x, y)|)] - [(1 - \mu) O_1(x, y) + \mu (1 - |O_1(x, y)|)] - x) \\
        &= -f(x, y) \\
        &= 0
    }

    \noindent
    Thus, we have shown that $f(-y, -x) = g(-y, -x) = 0$. This implies that the coordinate pair $\Delta Q^+ = -y$ and $\Delta Q^- = -x$ is a fixed point.
    
\end{proof}

\noindent
\newline
Using these two claims, we can make the following statement.

\begin{restatable}[(Bifurcation plots for $\Delta Q^+$ and $\Delta Q^-$ are equivalent)]{lemma}{bifurcationsymmetric}
\label{lemma:bifurcationsymmetric}

The bifurcation diagrams for $\Delta Q^+$ and $\Delta Q^-$ as a function of $\mu$ are both symmetric across the value $0$. Furthermore, the bifurcation diagrams will be equivalent.

\end{restatable}

\begin{proof}
    For a given $\mu = \mu^*$, suppose that $\Delta Q^+ = x$ and $\Delta Q^- = y$ is a fixed point. Using Claim \ref{claim:symmetric1} this will imply that $\Delta Q^+ = y$ and $\Delta Q^- = x$ is also a fixed point. Finally, by using Claim \ref{claim:symmetric2} on the previous two fixed points we can determine that $\Delta Q^+ = -y$ and $\Delta Q^- = -x$ and also $\Delta Q^+ = -x$ and $\Delta Q^- = -y$ are also fixed points. 

    \noindent
    \newline
    Note this implies the bifurcation diagram for $\Delta Q^+$ at $\mu$ will have the values $\Delta Q^+ \in \{-y, -x, x, y\}$. Similarly, the bifurcation diagram for $\Delta Q^-$ at $\mu^*$ will have the values $\Delta Q^- \in \{-y, -x, x, y\}$. Given that $\mu^*$ and the fixed point $\Delta Q^+ = x$ and $\Delta Q^- = y$ were arbitrary, we conclude that the bifurcation diagrams for both $\Delta Q^+$ and $\Delta Q^-$ are symmetric across $0$. Furthermore, the diagrams will be equivalent.
\end{proof}

\end{document}